\documentclass[11pt,reqno]{amsart}
\usepackage{amssymb,amsthm,amsmath,enumerate}
\usepackage{color,hyperref,url}
\usepackage[numbers,sort&compress]{natbib}
\usepackage{graphicx}
\usepackage{tikz}
\numberwithin{equation}{section}

\title[perturbed periodic Schr\"odinger operators]{On the perturbed periodic Schr\"odinger operators with separate resonant embedded eigenvalues}

\author{Kang Lyu}
\address{ School of Mathematics and Statistics, Nanjing University of Science and Technology, Nanjing, 210094, Jiangsu, People’s Republic of China}
\email{lvkang201905@outlook.com}
\urladdr{https://kanglyu.github.io/index.html}

\author{Chuanfu Yang}
\address{ School of Mathematics and Statistics, Nanjing University of Science and Technology, Nanjing, 210094, Jiangsu, People’s Republic of China}
\email{chuanfuyang@njust.edu.cn}

\keywords{Periodic Schr\"odinger operators,  absolutely continuous bands, embedded eigenvalues, generalized Pr\"ufer transformation.}

\thanks{{\em 2020 Mathematics Subject Classification.} Primary: 47E05. Secondary: 34L05}

\newcommand{\abs}[1]{\left\lvert #1 \right\rvert}

\theoremstyle{plain}
\newtheorem{theorem}{Theorem}[section]

\newtheorem{lemma}[theorem]{Lemma}
\newtheorem{proposition}[theorem]{Proposition}

\newcommand{\R}{\mathbb{R}}

\theoremstyle{definition}

\newtheorem{remark}[theorem]{Remark}

\begin{document}

\begin{abstract}In this paper, we consider Schr\"odinger operators on $L^2(0,\infty)$ given by
	\begin{align}
		Hu=(H_0+V)u=-u^{\prime\prime}+V_0u+Vu,\nonumber
	\end{align}
	where $V_0$ is real, $1$-periodic and $V$ is the perturbation. It is well known that under perturbations $V(x)=o(1)$ as $x\to\infty$, the essential spectrum of $H$ coincides with the essential spectrum of $H_0$. We  introduce a new way to construct oscillatory decaying perturbations with resonant embedded eigenvalues. Given any at most countable set $S$ inside the essential spectrum, we can construct perturbations with $S$ contained in the set of eigenvalues if the resonant eigenvalues in $S$ satisfy some condition. 
In particular, if $S$ is a finite set (or countable set), we can construct perturbation with $V(x)=\frac{O(1)}{x}$  $\left(\mathrm{or}\ \abs{V(x)}\leq\frac{h(x)}{1+x}\right)$ as $x\to\infty$ if the resonant eigenvalues of $S$ appear in the same spectral bands or large separate spectral bands, where $h(x)$ is any given function with $\lim_{x\to\infty}h(x)=\infty$.
\end{abstract}
	\maketitle
\section{introduction}
Let $H$ denote the one-dimensional Schr\"odinger operator defined on the half-axis by
differential expression
\begin{align}\label{perturbedperiodicschrodinger}
	Hu=(H_0+V)u=-u^{\prime\prime}+V_0u+Vu
\end{align}
and some self-adjoint boundary condition at zero, 
where $V_0\in L^2_{loc}(0,\infty)$ is real, $1$-periodic potential and $V$ is the perturbation. The operator describes a
charged particle, such as an electron, in the electric field $V_0+V$.
When $V=0$, we have an unperturbed $1$-periodic Schr\"odinger equation
\begin{align}\label{periodicschrodinger}
H_0u=-u^{\prime\prime}+V_0u=Eu.
\end{align}
It is well known that the essential spectrum of the operator induced by $H_0$ consists of absolutely continuous bands (no embedded eigenvalues). When $V(x)$ is decaying
quickly, one expects the spectral and dynamical properties of $H$ to remain close to
those of the operator $H_0$. In this paper, the perturbation $V$ always satisfies
\begin{align}
	V(x)=o(1),\nonumber
\end{align}
as $x\to\infty$. By Weyl's criterion, we know that $\sigma_{ess}(H)=\sigma_{ess}(H_0)$. However, the quality of the spectrum and dynamics may change. 
If the potential is small enough near $\infty$ (smaller than $\frac{o(1)}{1+x}$), then eigenvalues will not arise in these bands (\cite{LO17,kato1959}).
However, a little larger perturbation $\left(\frac{O(1)}{1+x}\right)$ may create embedded eigenvalues.
For example, for the special case $V_0=0$,  Wigner-von Neumann type potentials
\begin{align}
	V(x)=\frac{a}{1+x}\sin \left(kx+\theta\right)\nonumber
\end{align}
provide examples that $V(x)=\frac{O(1)}{1+x}$ and $H$ has a positive
embedded eigenvalue (\cite{von1929}).  From a physical perspective, this is a purely quantum resonance phenomenon, the appearance of an embedded eigenvalue suggests that at a particular point, an electron possesses enough energy to propagate but remains confined (\cite{kuchment00absenceofembedded}). The problems on the critical rates of decay at which some changes
in the spectral and dynamical properties of $H$ may happen have been widely studied these years. Such problems as the (sharp) spectral transitions for purely absolutely continuous spectrum, embedded singular continuous spectrum, embedded point spectrum (\cite{remling1998absolutely,remlingcouterexamples,remlingproceedingsboundsonsingular,kiselevsingularJAMS,liunonlineariterevisiting,liudiscrete,kiselevabsolutelyJAMS,kiselevstabilityofacDUKE,kato1959,liu2024sharpdecayperiodic,VoiculescuCMPabsolutelccontinuousKtheory}, etc.), and the estimate of ground-state energy (\cite{fefferman90bamsOnenergyofatom,fefferman92AIM.E.V.&E.F.ofODE,fefferman93aperiodicity,fefferman94AIMDiracandSchwinger,fefferman94AIMeigenvaluesumthreedimensional,fefferman94AIMthedensityonedimensional,fefferman94AIMthedensitythreedimensional}), minimization/maximization of the positive eigenvalues,  sharp ratios of the eigenvalues have been studied for different models (\cite{atkinson1978bounds,chuadvanceminimizationperiodic,ashbaughannalssharpratio,ashbaughcmpoptimalbounds}).

The problem of embedded eigenvalues into the essential spectrum or absolutely continuous spectrum arises from the special case that $V_0=0$. Here the essential spectrum of $H_0$ is purely absolutely continuous, and $\sigma_{ess}(H_0)=[0,\infty)$. As we mentioned, Wigner-von Neumann type potentials
provide with examples for one single embedded eigenvalue. This type of potentials were widely studied and developed later (\cite{Lukicjspwignervonneumann,Lukictransperiodic,Lukiccmpdecayingoscillatory,judgejacobineumann,janasdiscreteschrodinger,simonovzerosspectraldensity,simonovspectralanalysis}), and even to create the optimal bound for one single embedded eigenvalue (\cite{atkinson1978bounds,liudiscrete}). For many embedded eigenvalues, there are many fertile results these years \cite{Na86,simondense,marlettaembdedeigenvalues,kiselevperiodicmethod,liuirreducibilityfermivariety,krugerembeddedeigenvalues,liustarkmathN,liuasymptotic,kiselevsingularJAMS,vishwamdirac,judgegeometricapproach,judgediscretelevinsontechnique,JL19,Liuresonant,LO17,L19stark,liudiscrete}. For example, Naboko \cite{Na86} and Simon \cite{simondense} provided examples with countably many embedded eigenvalues. Judge, Naboko, and Wood \cite{judgegeometricapproach,judgediscretelevinsontechnique} introduced perturbations that create embedded eigenvalues for Jacobi operators.
In \cite{JL19}, 
Jitomirskaya and Liu proposed a new idea to create manifolds with (countably) many embedded eigenvalues, and this approach was developed later to solve different problems \cite{Liuresonant,LO17,L19stark}. 

Let us go back to our problem that the periodic term $V_0\neq 0$. Denote the absolutely continuous spectral bands by 
\begin{align}
	\sigma_{ess}(H_0)=\cup_{j=1}^\infty I_j,\nonumber
\end{align}
where $I_j\ (j\geq 1)$  are some closed intervals that will be introduced in Section 2, any two of these bands can intersect at most at one point.
For any $E\in I_j$, denote the corresponding Floquet solution by $\varphi(x,E)$, which satisfies
\begin{align}
	\varphi(x+1,E)=e^{ik(E)}\varphi(x,E),\nonumber
\end{align}
where $k=k(E)\in [0,\pi]$  is the quasimomentum. By Floquet theory, $k(E)$ monotonically decreases from $\pi$ to $0$ or monotonically increases from $0$ to $\pi$ in each $I_j$ (\cite{brown2013periodicbook}). Difficulties arise when one wants to create embedded eigenvalues. Point spectra are in a
sense very fragile. When modifying a perturbation $V(x)$ to create one eigenvalue often destroys other eigenvalues. Thus to construct a perturbation to have more than one embedded eigenvalue is very
challenging, let alone infinitely many. Moreover, the existence of  resonant eigenvalues results in the periodic case is extremely hard than the special case $V_0=0$.

 First of all, in each band, for every fixed  $k\in (0,\pi),k\neq \frac{\pi}{2}$, there are two eigenvalues $E_1$ and $E_2$ with $k(E_1)+k(E_2)=\pi$ will induce resonance. Namely, when one of the eigensolutions of \eqref{perturbedperiodicschrodinger} under a perturbation $V$, for example $u(x,E_1)$,  decreases quickly to be $L^2$ near $\infty$, the other one $u(x,E_2)$ highly possibly increases near $\infty$. Moreover, it will induce infinitely many couples of resonances by moving $k\in(0,\pi)$ even in a fixed band.
 Therefore, the first difficulty to overcome is how to construct a ``satisfactory" perturbation to create such kind of couples of resonant embedded eigenvalues from the same band. 
 What is more difficult is that the eigenvalues that can create resonances with some $E_1$ are not just $E_2$ from the same band but from infinitely many bands. In fact, suppose $ k(E_1)\in (0,\pi)$, then there are eigenvalues $E^k_j$ and $E^{\tilde{k}}_j$ with $k(E^k_j)=\pi-k(E^{\tilde{k}}_j)=k(E_1)$ in each band $I_j$. All these eigenvalues will create resonances with $E_1$.
 To avoid these difficulties, the previous papers always require a non-resonance assumption of the eigenvalues.
  For example, in \cite{Lukictransperiodic}, the non-resonance condition is addressed in their Lemma 13 expressed as a condition on Fourier coefficients.
 The authors in \cite{kiselevperiodicmethod} introduced a perturbation that has embedded eigenvalues with rationally independent quasimomenta. Later, W. Liu and his collaborator \cite{LO17} improved the method in \cite{JL19} and constructed perturbations step by step to create embedded eigenvalues under the non-resonance assumption. 
 
Inspired by \cite{JL19,kiselevsingularJAMS}, we introduce a new construction for perturbations. Our construction is more efficient in the aspect of the decreasing of eigensolutions than those of in \cite{JL19,LO17,L19stark,vishwamdirac,kangperiodicdirac}. Namely, the eigensolutions decrease continuously, not piecewise.
Once we take the eigenvalue $E_j$ into consideration in $j$-th step, the eigensolution starts decreasing as $x\to\infty$. 
This ensures that we are able to deal with some cases of resonances. Our main results are {Theorems} \ref{theoremmain124} and \ref{theoremmain224}. In our paper, the embedded eigenvalues are allowed to be resonant in the following two cases: \textbf{I}. The resonant eigenvalues appear in the same band. \textbf{II}. The resonant eigenvalues appear in many different (separate enough) bands (see formula \eqref{conditionb30}). We emphasize that even though there is no eigenvalues creating resonance with $E$ for $k(E)=\frac{\pi}{2}$ from the same band,  the previous papers do not permit the existence of the embedded eigenvalue $E$. In our paper, every $E$ with $k(E)=\frac{\pi}{2}$ could become an embedded eigenvalue under the new construction of the perturbations.
The definition of the perturbation is based on solving a class of differential equations involved with Pr\"ufer angles $\theta(x,E_j)$ which in turn depend on $V$, this idea was first came up with by Kiselev in \cite{kiselevsingularJAMS} and widely used in later papers \cite{LO17,Liuresonant}, etc.

Our paper is organized as follows. In Section 2, we give some asymptotics of eigenvalues and solutions of Sturm-Liouville equations and introduce main results. In Section 3, we first recall the generalized Pr\"ufer transformation, and then some important estimates are obtained.  In Section 4, we introduce a central quantitative estimate of an oscillatory integral coupled with a periodic parameter. In Section 5, we  prove some resonant integral estimates. In Section 6, we introduce the new construction and prove the main theorems.

\

\section{estimates of the eigenvalues and solutions of Sturm-Liouville equation and main results}

In this section we collect some properties of eigenvalues and solutions of the Sturm-Liouville equation, which we need in the following.

Define $\tilde{H}_k=\tilde{H}_k(V_0)$ by
\begin{align}\label{finiteintervalquasiperiodic}
	\tilde{H}_ku=-u^{\prime\prime}+V_0u=Eu,\ x\in (0,1),
\end{align}
with the boundary condition
\begin{align}\begin{cases}\label{quasiboundarycondition}
		u(1)=e^{ik}u(0),\\
		u'(1)=e^{ik}u'(0),
	\end{cases}
\end{align}
where $k\in [0,\pi]$. The boundary condition \eqref{quasiboundarycondition} is called periodic, antiperiodic or quasiperiodic boundary condition if $k=0$, $k=\pi$ or $k\in (0,\pi)$, respectively.

The eigenvalues of  $\tilde{H}_k$ (counted with multiplicity) are denoted by $E^k_1\leq E^k_2\leq E^k_3\leq \cdots.$ Then $E$ is an eigenvalue of $\tilde{H}_k$ if and only if 
\begin{align}
	D(E):=C(1,E)+S'(1,E)=2\cos k,
\end{align}
where $C(x,E)$ and $S(x,E)$ are solutions of \eqref{finiteintervalquasiperiodic} under the initial conditions
\begin{align}
	C(0,E)-1=C'(0,E)=S(0,E)=S'(0,E)-1=0.\nonumber
\end{align}  
Recall that the essential spectrum of $H_0$ is $\sigma_{ess}(H_0)=\cup_{j=1}^\infty I_j.$ By Floquet theory, 
we have for every $k\in (0,\pi)$, 
\begin{align}
	E^0_1<E^k_1<E^\pi_1\leq E^\pi_2<E^k_2<E^0_2\leq E^0_3<E^k_3<E^\pi_3\leq \cdots,
\end{align}
and $D(E)$ is analytic, strictly decreasing in $I_1=[E^0_1,E^\pi_1]$, $I_3=[E^0_3,E^\pi_3],\cdots$, strictly increasing in $I_2=[E^\pi_2,E^0_2]$, $I_4=[E^\pi_4,E^0_4],\cdots$. $k=k(E)$ increases in $I_1=[E^0_1,E^\pi_1]$, $I_3=[E^0_3,E^\pi_3],\cdots$, decreases in $I_2=[E^\pi_2,E^0_2]$, $I_4=[E^\pi_4,E^0_4],\cdots$.

Denote by $A=\int_0^1\abs{V_0(x)}dx$, $\rho=\sigma+i
\tau=\sqrt{E}$,
then one has the following classical asymptotics.
\begin{lemma}\cite[Lemma 1.1.2, Remark 1.1.2]{yurkobook}\label{lemmayurko}
	For any $\abs{\rho}>\max\{1,2A\}$, one has for any $x\in [0,1]$, 
	\begin{align}
		\label{Ccosrhox}\abs{C(x,E)-\cos\rho x}&\leq \frac{2Ae^{\abs{\tau}x}}{\abs{\rho}},\\
		\label{Cprimerhosinrho}\abs{C'(x,E)+\rho\sin\rho x}&\leq 2Ae^{\abs{\tau}x},
	\end{align}
	and
	\begin{align}
		\label{Ssinrhoxtorho}\abs{S(x,E)-\frac{\sin\rho x}{\rho}}&\leq \frac{2Ae^{\abs{\tau}x}}{\abs{\rho}^2},\\
		\label{Sprimecosrhox}\abs{S'(x,E)-\cos\rho x}&\leq \frac{2Ae^{\abs{\tau}x}}{\abs{\rho}}.
	\end{align}
\end{lemma}

Applying this lemma, for any fixed $k\in (0,\pi)$, the asymptotics of the eigenvalues can be obtained (see, for example \cite{yangquasitracJPA})
$\sqrt{E_{\pm n}}=2n\pi\pm k+\frac{O(1)}{n}$ as $n\to\infty$ (under an appropriate subscript).
In our paper, 
a quantitative description of $O(1)$ with respect to $n$ and $k$ is necessary. In the following, for any $k\in (0,\pi)$ and any $n\geq 1$, we denote by 
$$\delta_n(k)=\frac{Ae^{3}}{n\left(\sin k\sin\frac{99k}{100}\sin\frac{\pi+99k}{100}\right)}.$$

\begin{lemma}\label{lemmaeknakn}
	{	For any $k\in (0,\pi)$, we have for any $n> 1+\frac{A}{\frac{k(\pi-k)}{10^4}\sin k\sin\frac{99k}{100}\sin\frac{\pi+99k}{100}}$, 
		\begin{align}
			\label{eknakndeltank}\abs{\sqrt{E^k_n}-a^k_n}\leq \delta_n(k),
	\end{align}}
	where 
	\begin{align}
		a^k_n=\begin{cases}
			n\pi-k,&n\ is\ even,\\
			(n-1)\pi+k, &n\ is\ odd.\nonumber
		\end{cases}
	\end{align}
\end{lemma}
\begin{proof}
	Since $C(x,E)$ and $S'(x,E)$ satisfy Volterra integral equations (\cite[Lemma 1.1.2, Remark 1.1.2]{yurkobook}),
	\begin{align}
		C(x,E)=\cos\rho x+\int_0^x\frac{\sin \rho(x-t)}{\rho}V_0(t)C(t,E)dt,\\
		S'(x,E)=\cos\rho x+\int_0^x\cos \rho (x-t)V_0(t)S(t,E)dt,
	\end{align}
	applying \eqref{Ccosrhox} and \eqref{Ssinrhoxtorho}, we can obtain
	\begin{align}
		C(1,E)=\cos\rho +\int_0^1\frac{\sin\rho (1-t)}{\rho}\cos\rho t V_0(t)dt+f_1(\rho),
	\end{align}
	and
	\begin{align}
		S'(1,E)=\cos \rho +\int_0^1\cos\rho(1-t)\frac{\sin\rho t}{\rho}V_0(t)dt+f_2(\rho),
	\end{align}
	where 
	\begin{align}
		\abs{f_i(\rho)} \leq \frac{2A^2e^{\abs{\tau}}}{\abs{\rho}^2},\ i=1,2.
	\end{align}
	Therefore, the eigenvalues of $\tilde{H}_k$ coincide with the zeros of 
	\begin{align}
		\Delta(E):&=D(E)-2\cos k\nonumber\\
		&=C(1,E)+S'(1,E)-2\cos k\nonumber\\
		&=2\cos\rho-2\cos k+\frac{\sin\rho}{\rho}\int_0^1V_0(t)dt+f(\rho),
	\end{align}
	where $\abs{f(\rho)}\leq \frac{4A^2e^{\abs{\tau}}}{\abs{\rho}^2}$. 
	
	For any $n\in\mathbb{N}$, denote by 
	\begin{align}
		\gamma_n^{\pm}(k)=&\{\rho\in\mathbb{C}:{\sigma}=k+2n\pi \pm\delta_n(k),\ {\abs{\tau}}\leq k+2n\pi\pm\delta_n(k)\}\nonumber\\
		&\cup\{\rho\in\mathbb{C}: {\abs{\tau}}= k+2n\pi\pm\delta_n(k),\ 0\leq {\sigma}\leq k+2n\pi \pm\delta_n(k)\}.\nonumber
	\end{align}
	Let
	\begin{align}
		\Gamma_n^{\pm}(k)=\{E=\rho^2:\rho\in \gamma_n^{\pm}(k)\} (\mathrm{\ see  \ Figure\ 1}).\nonumber
	\end{align}
	
\begin{center}
	\tikzset{every picture/.style={line width=0.75pt}} 
	
	\begin{tikzpicture}[x=0.70pt,y=0.70pt,yscale=-1,xscale=1]
		
		\draw  (0,167) -- (267,167)(23.75,4) -- (23.75,307) (260,162) -- (267,167) -- (260,172) (18.75,11) -- (23.75,4) -- (28.75,11)  ;
		\draw  (284,165.87) -- (600,165.87)(437.21,10) -- (437.21,312) (593,160.87) -- (600,165.87) -- (593,170.87) (432.21,17) -- (437.21,10) -- (442.21,17)  ;
		\draw    (233,95) .. controls (255.77,43.52) and (295.2,43.01) .. (319.28,95.4) ;
		\draw [shift={(320,97)}, rotate = 246.04] [color={rgb, 255:red, 0; green, 0; blue, 0 }  ][line width=0.75]    (10.93,-3.29) .. controls (6.95,-1.4) and (3.31,-0.3) .. (0,0) .. controls (3.31,0.3) and (6.95,1.4) .. (10.93,3.29)   ;
		\draw    (25,49) -- (220,49) ;
		\draw    (220,49) -- (221,279) ;
		\draw    (23,280) -- (221,279) ;
		\draw    (130,80) -- (130,250) ;
		\draw    (130,80) -- (24,80) ;
		\draw    (130,250) -- (24,251) ;
		\draw   (437.21,83.22) .. controls (285.88,139.89) and (286.28,195.47) .. (437.21,249.99) ;
		\draw   (437.21,83.5) .. controls (604.92,137.79) and (605.32,193.28) .. (437.21,249.99) ;
		\draw  [line width=1.5] [line join = round][line cap = round] (177,165.5) .. controls (177,165.5) and (177,165.5) .. (177,165.5) ;
		\draw   (437.21,71.98) .. controls (262.63,136.26) and (263.08,199.25) .. (437.21,261) ;
		\draw   (437.21,261) .. controls (627.12,198.69) and (627.35,135.68) .. (437.21,71.98) ;
		
		\draw (253,31) node [anchor=north west][inner sep=0.75pt]   [align=left] {$E=\rho^2$};
		\draw (88,116) node [anchor=north west][inner sep=0.75pt]   [align=left] {$\gamma_n^-(k)$};
		\draw (222,116) node [anchor=north west][inner sep=0.75pt]   [align=left] {$\gamma_n^+(k)$};
		\draw (155,168) node [anchor=north west][inner sep=0.75pt]   [align=left] {$k+2n\pi$};
		\draw (509,80) node [anchor=north west][inner sep=0.75pt]   [align=left] {$\Gamma_n^+(k)$};
		\draw (455,107) node [anchor=north west][inner sep=0.75pt]   [align=left] {$\Gamma_n^-(k)$};
		\draw (10,167) node [anchor=north west][inner sep=0.75pt]   [align=left] {$0$};
		\draw (425,167) node [anchor=north west][inner sep=0.75pt]   [align=left] {$0$};
		\draw (450,20) node [anchor=north west][inner sep=0.75pt]   [align=left] {$E-plane$};
        \draw (35,20) node [anchor=north west][inner sep=0.75pt]   [align=left] {$\rho-plane$};
	\end{tikzpicture}
\end{center}
\begin{center}
	$Fig.1. \ Contours\ \Gamma_n^{\pm}(k)$
\end{center}
	We show that for any $E\in\Gamma_n^{\pm}(k)$ with 
	$n>1+\frac{A}{\frac{k(\pi-k)}{10^4}\sin k\sin\frac{99k}{100}\sin\frac{\pi+99k}{100}},$ 
	there holds
	\begin{align}
		\label{cosrho-cosk}	\abs{2\cos\rho-2\cos k}> \abs{\frac{\sin\rho}{\rho}\int_0^1V_0(t)dt+f(\rho)}.
	\end{align}
	
	Firstly, for any $\rho\in\gamma^{\pm}_n(k)$ with $\abs{\tau}\geq 3$,
	we have 
	\begin{align}
			\abs{2\cos\rho-2\cos k}e^{-\abs{\tau}}&=\abs{e^{i\sigma- \tau}\nonumber +e^{-i\sigma+\tau}-2\cos k}e^{-\abs{\tau}}\\
		&\geq 1-3e^{-\abs{\tau}}\nonumber\\
		&>\frac{1}{2}.\nonumber
	\end{align}
	Hence one has that for any $\rho\in\gamma^{\pm}_n(k)$ with $\abs{\tau}\geq 3$,
	\begin{align}
		\abs{2\cos\rho-2\cos k}>\frac{1}{2}e^{\abs{\tau}}.\label{cosrho-cosk1}
	\end{align}
	For any $\rho$ with $\sigma=k+2n\pi \pm\delta_n(k),\ \abs{\tau}<3$, one has
	\begin{align}
		\label{cosrho-cosk2}	\abs{2\cos\rho-2\cos k}e^{-\abs{\tau}}\geq \abs{2\cos\rho-2\cos k}e^{-3},
	\end{align}
	and it is not difficult to obtain that 
	\begin{align}
		\abs{2\cos\rho-2\cos k}&\geq \abs{2\cos(k\pm \delta_n(k))-2\cos k}\nonumber\\
		\label{cosrho-cosk3}	&>2\delta_n(k)\sin k\sin\frac{99k}{100}\sin\frac{\pi+99k}{100}.
	\end{align}
	The later inequality uses the fact that 
	\begin{align}
		\delta_n(k)<\frac{k(\pi-k)e^3}{10^4}<\min\left\{\frac{k}{100},\frac{\pi-k}{100}\right\}.\nonumber
	\end{align}
	Therefore, by \eqref{cosrho-cosk2} and \eqref{cosrho-cosk3}, for any $\rho\in\gamma^{\pm}_n(k)$ with $\abs{\tau}< 3$, there holds
	\begin{align}
		\abs{2\cos\rho-2\cos k}>2\delta_n(k)\sin k\sin\frac{99k}{100}\sin\frac{\pi+99k}{100}e^{-3}e^{\abs{\tau}}
		=\frac{2A}{n}e^{\abs{\tau}}>\frac{2A}{\abs{\rho}}e^{\abs{\tau}}.\label{middle107}
	\end{align}
	Since $n>1+\frac{A}{\frac{k(\pi-k)}{10^4}\sin k\sin\frac{99k}{100}\sin\frac{\pi+99k}{100}}>4A$, 
	one obtains for any $E\in\Gamma_n^{\pm}(k)$,
	\begin{align}
		\abs{f(\rho)}\leq \frac{4A^2e^{\abs{\tau}}}{\abs{\rho}^2}<\frac{Ae^{\abs{\tau}}}{\rho}.\nonumber
	\end{align}
	Thus, for any $E\in\Gamma_n^{\pm}(k)$, 
	\begin{align}
		\label{cosrho-cosk4}\abs{\frac{\sin\rho}{\rho}\int_0^1V_0(t)dt+f(\rho)}\leq \frac{2Ae^{\abs{\tau}}}{\abs{\rho}}.
	\end{align}
	Then by \eqref{cosrho-cosk1}, \eqref{middle107} and \eqref{cosrho-cosk4} one can obtain \eqref{cosrho-cosk}. By similar calculations, one can obtain that for any
		\begin{align}
		E\in\Gamma_n^{\pm}(-k)=\{E=\rho^2:\rho\in \gamma_n^{\pm}(-k)\},\nonumber
	\end{align}
with  $n>1+\frac{A}{\frac{k(\pi-k)}{10^4}\sin k\sin\frac{99k}{100}\sin\frac{\pi+99k}{100}},$	
we have 
\begin{align}
	\label{cosrho-cosk5}	\abs{2\cos\rho-2\cos k}> \abs{\frac{\sin\rho}{\rho}\int_0^1V_0(t)dt+f(\rho)},
\end{align}where
	\begin{align}
		\gamma_n^{\pm}(-k)=&\{\rho\in\mathbb{C}:{\sigma}=-k+2n\pi \pm\delta_n(k),\ {\abs{\tau}}\leq -k+2n\pi \pm\delta_n(k)\}\nonumber\\
		&\cup\{\rho\in\mathbb{C}: {\abs{\tau}}= -k+2n\pi \pm\delta_n(k),\ 0\leq {\sigma}\leq -k+2n\pi \pm\delta_n(k)\}.\nonumber
	\end{align} 
	Therefore, by \eqref{cosrho-cosk} and \eqref{cosrho-cosk5}, applying Rouché's theorem, we can obtain the result.
	
\end{proof}

\begin{lemma}\label{CxSxasymptotics}
	For any $k\in (0,\pi)$ and for $n>1+\frac{A}{\frac{k(\pi-k)}{10^4}\sin k\sin\frac{99k}{100}\sin\frac{\pi+99k}{100}}$, we have for any $x\in [0,1]$,
	\begin{align}
		\label{Cx} -1-\delta_n(k)\leq 	C(x,E^k_n)\leq 1+\delta_n(k),\\
		\label{Sx}-1-\delta_n(k)\leq \sqrt{E^k_n}S(x,E^k_n)\leq	1+\delta_n(k),
	\end{align}
	\begin{align}
		\label{Csquaredx}\cos^2\left(\sqrt{E^k_n}x\right)-\delta_n(k)\leq C^2(x,E^k_n)\leq \cos^2\left(\sqrt{E^k_n}x\right)+\delta_n(k),\\
		\label{Ssquaredx}{\sin^2\left(\sqrt{E^k_n}x\right)}-\delta_n(k) \leq E^k_nS^2(x,E^k_n)\leq {\sin^2\left(\sqrt{E^k_n}x\right)}+\delta_n(k),
	\end{align}
	and
	\begin{align}
		\label{C1}\cos k-2\delta_n(k)&\leq C(1,E^k_n)\leq \cos k+2\delta_n(k),\\
		\label{Csquared1}\cos^2 k-2\delta_n(k)&\leq	C^2(1,E^k_n)\leq \cos^2 k+2\delta_n(k)\\
		\label{S1}(-1)^{n+1}{\sin k}-2\delta_n(k)&\leq \sqrt{E^k_n}S(1,E^k_n)\leq (-1)^{n+1}{\sin k}+2\delta_n(k),\\
		\label{Ssquared1}{\sin^2 k-2\delta_n(k)} &\leq E^k_nS^2(1,E^k_n)\leq {\sin^2 k+2\delta_n(k)}.
	\end{align}
\end{lemma}
\begin{proof}
	We only give the proofs of \eqref{Cx},\eqref{Csquaredx}, \eqref{C1} and \eqref{Csquared1}, the proofs of \eqref{Sx}, \eqref{Ssquaredx}, \eqref{S1} and \eqref{Ssquared1} are similar and omitted. 
	
	Firstly, by \eqref{eknakndeltank}, for any $n>1+\frac{A}{\frac{k(\pi-k)}{10^4}\sin k\sin\frac{99k}{100}\sin\frac{\pi+99k}{100}}$, one has that 
	\begin{align}
		\sqrt{E^k_n}>(n-1)\pi.\nonumber
	\end{align}
	Hence, one obtains
	\begin{align}
		\delta_n(k)>\frac{8A}{\sqrt{E^k_n}}.\label{deltankgeq8asqrtekn}
	\end{align}
	
	By \eqref{Ccosrhox}, one has 
	\begin{align}
		\cos\left(\sqrt{E^k_n}x\right)-\frac{2A}{\sqrt{E^k_n}}\leq C(x,E^k_n)\leq \cos \left(\sqrt{E^k_n}x\right)+\frac{2A}{\sqrt{E^k_n}},\label{Cxleqgeq}
	\end{align}
	hence, one can obtain \eqref{Cx} by \eqref{deltankgeq8asqrtekn}.
	
	By \eqref{Cxleqgeq} one can obtain 
	\begin{align}
		\abs{C^2(x,E^k_n)-\cos ^2\left(\sqrt{E^k_n}x\right)}=&\abs{C(x,E^k_n)-\cos \left(\sqrt{E^k_n}x\right)}\abs{C(x,E^k_n)+\cos \left(\sqrt{E^k_n}x\right)}\nonumber\\
		\leq & \frac{2A}{\sqrt{E^k_n}}\left(2+\frac{2A}{\sqrt{E^k_n}}\right)\nonumber\\
		=&\frac{4A^2}{E^k_n}+\frac{4A}{\sqrt{E^k_n}},\nonumber
	\end{align}
	then by $\sqrt{E^k_n}>n>A$ and \eqref{deltankgeq8asqrtekn}, we can obtain \eqref{Csquaredx}. 
	
	By \eqref{Cxleqgeq}, one has
	\begin{align}
		\cos\sqrt{E^k_n}-\frac{2A}{\sqrt{E^k_n}}\leq C(1,E^k_n)\leq \cos \sqrt{E^k_n}+\frac{2A}{\sqrt{E^k_n}}.\nonumber
	\end{align}
	Applying \eqref{eknakndeltank} and mean value theorem, we have
	\begin{align}
		\abs{\cos \sqrt{E^k_n}-\cos k}\leq \delta_n(k).\nonumber
	\end{align}
	Then one can obtain \eqref{C1} by \eqref{deltankgeq8asqrtekn}. 
	
	By \eqref{Csquaredx}, one has
	\begin{align}
		\cos^2\sqrt{E^k_n}-\delta_n(k)\leq C^2(1,E^k_n)\leq \cos^2\sqrt{E^k_n}+\delta_n(k).\nonumber
	\end{align}
	By \eqref{eknakndeltank} and mean value theorem, we can obtain
	\begin{align}
		\abs{\cos ^2\sqrt{E^k_n}-\cos^2 k}\leq \delta_n(k).\nonumber
	\end{align}
	Then one obtains \eqref{Csquared1}.
\end{proof}
Now we state our main results.
Denote by $I_1^\circ=(E^0_1,E^\pi_1)$, $I_2^\circ=(E^\pi_2,E^0_2)$, $I_3^\circ=(E^0_3,E^\pi_3),$ $I_4^\circ=(E^\pi_4,E^0_4),\cdots$, and $\sigma_{ess}^\circ=\cup_{j=1}^\infty  I^\circ_j$.
For any $E_j=E^{k(E_j)}_{n(E_j)}\in I^\circ_{n(E_j)}$ (the $n(E_j)$-th eigenvalue of $\tilde{H}_{k(E_j)}$), for the simplification of notations, we denote by $k_j=k(E_j)$, $n_j=n(E_j)$ if there is no confusion. Then $E_j=E^{k_j}_{n_j}$. We require that in the following, for any $k,k_j\in(0,\pi)$, $\tilde{k}=\pi-k,\tilde{k}_j=\pi-k_j$.
For any $k\in (0,\pi)$, define
\begin{align}
	L(k)&={1+\frac{A}{\frac{k(\pi-k)}{10^4}\sin^3 k\sin\frac{99k}{100}\sin\frac{\pi+99k}{100}}},\label{deLk}\\
	{\delta}(k)&=\frac{40\pi n\delta_n(k)}{\sin^2k}=\frac{40\pi Ae^3}{\sin^3 k\sin\frac{99k}{100}\sin\frac{\pi+99k}{100}},\label{dedeltak}
\end{align}
we have $L(k)=L(\tilde{k})$, ${\delta}(k)={\delta}(\tilde{k}).$

 Given any at most countable set of distinct reals
\begin{align}
	S:=\{E_j\}_{j=1}^K=\left\{E^{k_j}_{n_j}\right\}_{j=1}^{{K}}\subset\sigma_{ess}^\circ,\nonumber
\end{align}
where $1\leq K\leq \infty$. 
Let us assume that the elements in $S$ satisfies the assumption:

\noindent\textbf{A1} Assume that $S$ can be divided into $3$ disjoint subsets $S=S_1\cup S_2\cup S_3$. For any $E_j\in S_1$, one has for any $E_i\in S$ with $i\neq j$ that $k_i\neq k_j,\tilde{k}_j$. Namely, $S_1$ consists of the non-resonant eigenvalues in $S$. For any $E_j\in S_2$, one has $k_j\neq \frac{\pi}{2}$, and there exists $E_i\in S$ such that $k_i=\tilde{k}_j$ and $n_i=n_j$, and for any $l\neq i,j$, there holds $k_l\neq k_i,k_j$. In other words, $S_2$ consists of the resonant eigenvalues (only) from the same band. For any $E_j\in S_3$, one has $n_j>L(k_j)$, and denote $\{E_i\}_{i\in I}=\{E_i\in S:k_i=k_j\mathrm{\ or}\ k_i=\tilde{k}_j\}$ (of course $j\in I$), then $I$ has cardinality $\# I\geq 2$ and for any $i\in I$ that
\begin{align}
\frac{1}{2n_i}+\sum_{l\neq i,l\in I}\frac{n_l}{n_i\abs{n_i-n_l}}<\frac{1}{50\pi+50\pi{\delta}(k_j)}.\label{conditionb30}
\end{align}

\begin{theorem}\label{theoremmain124}
	Given any set of distinct reals
	$$S=\{E_j\}_{j=1}^K\subset \sigma_{ess}^\circ$$
	 with $K<\infty$, suppose that $S$ satisfies the assumption \textbf{A1}.
	Then for any $\{\xi_j\}_{j=1}^K\subset [0,\pi]$, there exist potentials with
\begin{align}
	V(x)=\frac{O(1)}{1+x},\nonumber
\end{align}
as $x\to\infty$ such that $Hu=E_ju$ has an $L^2(0,\infty)$ solution $u_j(x)$ with the boundary condition
\begin{align}
	\frac{u'_j(0)}{u_j(0)}=\tan \xi_j,\ j=1,2,\cdots,K.\nonumber
\end{align}
\end{theorem}

\begin{theorem}\label{theoremmain224}
	Given any countable set of distinct reals 
	$$S=\{E_j\}_{j=1}^\infty\subset \sigma_{ess}^\circ,$$ suppose that $S$ satisfies the assumption \textbf{A1}.
Let $h(x)>0$ be any function on $(0,\infty)$ with $\lim_{x\to\infty}h(x)=\infty$. Then for any $\{\xi_j\}_{j=1}^\infty\subset [0,\pi]$, there exist potentials with $V(x)=o(1)$ as $x\to\infty$ and
	\begin{align}
	\abs{V(x)}\leq \frac{h(x)}{1+x},\ \ for\ any\ x>0,\nonumber
	\end{align}
such that $Hu=E_ju$ has an $L^2(0,\infty)$ solution $u_j(x)$ with the boundary condition
	\begin{align}
		\frac{u'_j(0)}{u_j(0)}=\tan \xi_j,\ j=1,2,\cdots.\nonumber
	\end{align}
\end{theorem}

\section{modified pr\"ufer transformation}
In this section, we first recall the generalized
Pr\"ufer transformation of Schr\"odinger equation $Hu=Eu$, which is from \cite{kiselevperiodicmethod}, then we give some properties for the Pr\"ufer variables.

For any $k\in (0,\pi)$, denote the corresponding Floquet solution to the $n$-th eigenvalue $E=E^k_n\in I^\circ_n$ by $\varphi(x,E)$.  It is well known that (\cite{brown2013periodicbook}) 
\begin{align}
	\varphi(x,E)=p(x,E)e^{ikx},\label{varphixpx}
\end{align}
where $p(x,E)$ is a $1$-periodic function in $x$.

Define the Wronskian of $f$ and $g$ by
\begin{align}
	W(f,g)(x)=f(x)g'(x)-f'(x)g(x).\nonumber
\end{align}
Then $W(\bar{\varphi}(x,E),\varphi(x,E))=2i{\rm{Im}}(\bar{\varphi}(x,E)\varphi'(x,E))\neq 0$ is a constant. 
For $E=E^k_n$, denote by
\begin{align}
	\omega^k_n=2{\rm{Im}}(\bar{\varphi}(x,E)\varphi'(x,E)).\label{definitionofomegakn}
\end{align}
Define $\eta(x,E)$ as a continuous function by
\begin{align}
	\varphi(x,E)=\abs{\varphi(x,E)}e^{i\eta(x,E)},\label{varphieta}
\end{align}
and $\eta(0,E)\in (-\pi,\pi]$. Clearly, by \eqref{varphixpx}, one has for some $\gamma(x,E)$,
\begin{align}
	\eta(x,E)=kx+\gamma(x,E),\label{etakegamma}
\end{align}
where $\gamma(x,E) \mod2\pi$ is a function that is $1$-periodic in $x$.
Let $u(x,E)$ be a real solution of \eqref{perturbedperiodicschrodinger}. Define $\rho(x,E)$, $R(x,E)$ and $\theta(x,E)$ by 
\begin{align}
	\begin{pmatrix}
		u(x,E)\\
		u'(x,E)
	\end{pmatrix}={\rm{Im}}\left[\rho(x,E)\begin{pmatrix}
		\varphi(x,E)\\
		\varphi'(x,E)
	\end{pmatrix}
	\right],\nonumber
\end{align}
and 
\begin{align}
	R(x,E)&=\abs{\rho(x,E)},\nonumber\\
	\theta(x,E)&={\rm{Arg}}(\rho(x,E))+\eta(x,E).\nonumber
\end{align}
Then one has the following proposition.
\begin{proposition}\cite[Prop. 2.1(b), 2.2 and Thm. 2.3]{kiselevperiodicmethod}
	For any $k=k(E)\in (0,\pi)$ and $E=E^k_n\in I^\circ_n$, there exists $C=C(E)>0$, such that
	\begin{align}
		\frac{1}{C}\left(u(x,E)^2+u'(x,E)^2\right)\leq R(x,E)^2\leq C\left(u(x,E)^2+u'(x,E)^2\right),\label{usquaredrsquared}
	\end{align}
	and
	\begin{align}\label{etaprime}
		\eta'(x,E)=\frac{\omega^k_n}{2\abs{\varphi(x,E)}^2},
	\end{align}
	\begin{align}\label{pruferR24}
		(\ln R(x,E))'=\frac{V(x)}{2\eta'(x,E)}\sin2\theta(x,E),
	\end{align}
	\begin{align}\label{prufertheta24}
		\theta'(x,E)=\eta'(x,E)-\frac{V(x)}{\eta'(x,E)}\sin^2\theta(x,E).
	\end{align}
\end{proposition}

We mention that there is a little difference between the definitions of $\omega^k_n$ in our paper and that in \cite{kiselevperiodicmethod}. In \cite{kiselevperiodicmethod}, they always require that $\omega^k_n>0$ by interchanging $\varphi$ and $\bar{\varphi}$. However, in our paper, for convenience in the final proof, we do not ask $\omega^k_n$ always to be positive (see Remark \ref{omegaknpm}). Instead, we normalize $\varphi(x,E)$ by $\varphi(0,E)=1$.

Clearly, by \eqref{usquaredrsquared}, we know that $u(\cdot,E)\in L^2(0,\infty)$ is equivalent to $R(\cdot,E)\in L^2(0,\infty).$ To obtain the asymptotics of $R(x,E)$ near $\infty$, we need some properties of $\eta'(x,E)$.
Since $\varphi(x,E)\neq 0$ (\cite[Lemma 8]{stolzperiodic}), we may normalize $\varphi(x,E)$ by requiring $\varphi(0,E)=1$. Then
we can assume that 
\begin{align}
	\varphi(x,E)=C(x,E)+bS(x,E).\nonumber
\end{align}
By \eqref{quasiboundarycondition} we can obtain
\begin{align}
	\varphi(x,E)=&C(x,E)+\frac{e^{ik}-C(1,E)}{S(1,E)}S(x,E)\nonumber\\
	=&C(x,E)+\frac{\cos k-C(1,E)}{S(1,E)}S(x,E)+i\frac{ S(x,E)\sin k}{S(1,E)}.\label{varphixkncxeknsxekn}
\end{align}

\begin{proposition}
	For any $k=k(E)\in (0,\pi)$ and $E=E^k_n\in I^\circ_n$, we have
	\begin{align}\label{omegaknsink}
		\omega^k_n={\frac{2\sin k}{{S(1,E)}}}.
	\end{align}
\end{proposition}
\begin{proof}
	By \eqref{definitionofomegakn} and \eqref{varphixkncxeknsxekn}, one has
	\begin{align}
		\omega^k_n
		=&2\Bigg(\left(C(x,E)+\frac{\cos k-C(1,E)}{S(1,E)}S(x,E)\right)\frac{ S'(x,E)\sin k}{S(1,E)}-\nonumber\\
		&\left(C'(x,E)+\frac{\cos k-C(1,E)}{S(1,E)}S'(x,E)\right)\frac{ S(x,E)\sin k}{S(1,E)}\Bigg)
		\nonumber\\
		=&\frac{2\sin k}{S(1,E)}.\nonumber
	\end{align}
\end{proof}

\begin{remark}\label{omegaknpm}
	By classical Sturm-Liouville theory, we know that for any $n\in\mathbb{N}$, $(-1)^nS(1,E^k_n)<0$ (see, for example \cite{yurkobook,poschelbookinverse}). Therefore, for any odd $n$, we have $\omega^k_n>0$ and for any even $n$, we have $\omega^k_n<0$.
\end{remark}

\begin{lemma}\label{lemmaomegaovervarphisqaured}
	Let $k=k(E)\in (0,\pi)$ and $E=E^k_n\in I^\circ_n$ with $n>L(k)$, then we have for any $x\geq 0$,
	\begin{align}\label{omegavarphiestimate}
		(n-1)\pi-{\delta}(k)\leq (-1)^{n+1}\eta'(x,E) \leq n\pi+{\delta}(k).
	\end{align}
\end{lemma}
\begin{proof} Since $\eta'(x,E)$ is $1$-periodic in $x$, we assume that $x\in [0,1]$.
	By \eqref{varphixkncxeknsxekn} we have 
	\begin{align}
		\abs{\varphi(x,E)}^2=&\abs{C(x,E)+\frac{\cos k-C(1,E)}{S(1,E)}S(x,E)}^2+\abs{\frac{ S(x,E)\sin k}{S(1,E)}}^2\nonumber\\
		=&C^2(x,E)+\frac{2\left(\cos k-C(1,E)\right)C(x,E)S(x,E)}{S(1,E)}\label{varphisquqred24}\\
		&+\frac{\left(\cos k-C(1,E))^2S^2(x,E\right)}{S^2(1,E)}+\sin ^2k\frac{S^2(x,E)}{S^2(1,E)}.\nonumber
	\end{align}
	Applying Lemma \ref{CxSxasymptotics}, one has
	\begin{align}
		\label{absvarphisquaredleq}	\abs{\varphi(x,E)}^2\leq& \cos^2\left(\sqrt{E}x\right)+\delta_n(k)+\frac{4\delta_n(k)(1+\delta_n(k))^2}{\sin k-2\delta_n(k)}\\
		&+\frac{4\delta_n(k)^2(1+\delta_n(k))}{\sin^2k-2\delta_n(k)}+\sin^2k\frac{\sin^2\left(\sqrt{E}x\right)+\delta_n(k)}{\sin ^2k-2\delta_n(k)}.\nonumber
	\end{align}
	By $n>{1+\frac{A}{\frac{k(\pi-k)}{10^4}\sin^3 k\sin\frac{99k}{100}\sin\frac{\pi+99k}{100}}}>\frac{32Ae^3}{\sin^3k\sin\frac{99k}{100}\sin\frac{\pi+99k}{100}}$, one can obtain
	\begin{align}
		\frac{2\delta_n(k)}{\sin k}\leq \frac{2\delta_n(k)}{\sin ^2k}=\frac{2Ae^3}{n\sin^3k\sin\frac{99k}{100}\sin\frac{\pi+99k}{100}}<\frac{1}{16}.\label{oneto16}
	\end{align}
	Hence, 
	\begin{align}
		\frac{1}{\sin k-2\delta_n(k)}\leq \frac{1}{\sin k}\left(1+\frac{4\delta_n(k)}{\sin k}\right),\label{finalsink}\\
		\frac{1}{\sin^2 k-2\delta_n(k)}\leq \frac{1}{\sin^2 k}\left(1+\frac{4\delta_n(k)}{\sin^2 k}\right).\label{finalsink^2}
	\end{align}
	Therefore, by \eqref{absvarphisquaredleq}, one has
	\begin{align}
		\abs{\varphi(x,E)}^2\leq& \cos^2 \left(\sqrt{E}x\right)+\delta_n(k)+\frac{8\delta_n(k)}{\sin k}+\frac{8\delta_n(k)^2}{\sin^2k}+\sin^2\left(\sqrt{E}x\right)
		+\delta_n(k)+\frac{5\delta_n(k)}{\sin^2 k}\nonumber\\
		\leq & 1+\frac{16}{\sin^2k}\delta_n(k).\label{varphixleq1}
	\end{align}
	Similarly, we can obtain 
	\begin{align}
		\abs{\varphi(x,E)}^2\geq 1-\frac{16}{\sin^2k}\delta_n(k).\label{varphixlargethan1-16}
	\end{align}
	Then by Remark \ref{omegaknpm}, \eqref{S1} and \eqref{omegaknsink}, one obtains
	\begin{align}
		{\frac{(-1)^{n+1}\omega^k_n}{2\abs{\varphi(x,E)}^2}}=\frac{\sin k}{\abs{S(1,E)}\abs{\varphi(x,E)}^2}\leq \frac{\sqrt{E}}{\left(1-\frac{2\delta_n(k)}{\sin k}\right)\left(1-\frac{16\delta_n(k)}{\sin^2k}\right)}.\nonumber
	\end{align}
	Using \eqref{oneto16}, one has
	\begin{align}
		{\frac{(-1)^{n+1}\omega^k_n}{2\abs{\varphi(x,E)}^2}}\leq \sqrt{E}\left(1+\frac{40}{\sin^2k}\delta_n(k)\right).\nonumber
	\end{align}
	Applying
	\eqref{S1}, \eqref{omegaknsink}, \eqref{oneto16} and  \eqref{varphixleq1}, one can obtain
	\begin{align}
		{\frac{(-1)^{n+1}\omega^k_n}{2\abs{\varphi(x,E)}^2}}\geq \sqrt{E}\left(1-\frac{40}{\sin^2k}\delta_n(k)\right).\nonumber
	\end{align}
	Then by \eqref{eknakndeltank} and \eqref{etaprime} one can obtain \eqref{omegavarphiestimate}.
\end{proof}

\begin{lemma}Let $k=k(E)\in (0,\pi)$ and $E=E^k_n\in I^\circ_n$ with $n>L(k)$, then we have for any $x\geq 0$,
	\begin{align}
		\abs{\eta^{\prime\prime}(x,E)}\leq n\pi {\delta}(k).\label{etaprimeprime}
	\end{align}
\end{lemma}
\begin{proof}
	By \eqref{varphisquqred24}, one has
	\begin{align}
		\left(\abs{\varphi(x,E)}^2\right)'=&2C(x,E)C'(x,E)+\frac{2(\cos k-C(1,E))}{S(1,E)}\nonumber\\
		&\times(C(x,E)S'(x,E)+C'(x,E)S(x,E))\nonumber\\
		&+2\sin^2k\frac{S'(x,E)S(x,E)}{S^2(1,E)}+\frac{2(\cos k-C(1,E))^2S'(x,E)S(x,E)}{S^2(1,E)}\nonumber\\
		:=&q_1+q_2+q_3+q_4.\nonumber
	\end{align}	
	By \eqref{Ccosrhox} and \eqref{Cprimerhosinrho}, one has
	\begin{align}
		\abs{q_1+\sqrt{E}\sin\left(2\sqrt{E}x\right)}\leq&\abs{2C(x,E)C'(x,E)-2\cos\left(\sqrt{E}x\right)C'(x,E)}\nonumber\\
		&+\abs{2\cos\left(\sqrt{E}x\right)C'(x,E)+\sqrt{E}\sin\left(2\sqrt{E}x\right)}\nonumber\\
		\leq &12A.\nonumber
	\end{align}
	Applying \eqref{Ccosrhox}-\eqref{Sprimecosrhox}, and \eqref{C1}, \eqref{S1}, one can obtain 
	\begin{align}
		\abs{q_2}\leq \frac{4\delta_n(k)}{\sin k-2\delta_n(k)}\left(\sqrt{E}+12A\right).\nonumber
			\end{align}
	For $q_3$ we have
	\begin{align}
		\Bigg|q&_3-\sqrt{E}\sin\left(2\sqrt{E}x\right)\Bigg|\nonumber\\
		=&\abs{\frac{2\sin^2k}{ES^2(1,E)}\left(ES'(x,E)S(x,E)-\frac{ES^2(1,E)}{2\sin^2k}\sqrt{E}\sin\left(2\sqrt{E}x\right)\right)}\nonumber\\
		\leq &\abs{\frac{2\sin^2k}{ES^2(1,E)}\left(ES'(x,E)S(x,E)-\sqrt{E}S'(x,E)\sin\left(\sqrt{E}x\right) \right)}\nonumber\\
		&+\abs{\frac{2\sin^2k}{ES^2(1,E)}\sqrt{E}\sin\left(\sqrt{E}x\right)\left(S'(x,E)-\cos\left(\sqrt{E}x\right)\right)}\nonumber\\
		&+\abs{\frac{2\sin^2k}{ES^2(1,E)}\sqrt{E}\sin\left(\sqrt{E}x\right)\cos\left(\sqrt{E}x\right)\left(1-\frac{ES^2(1,E)}{\sin^2k}\right)}.\nonumber
	\end{align}	
	By \eqref{Ssinrhoxtorho}, \eqref{Sprimecosrhox} and \eqref{Ssquared1}, one has
	\begin{align}
		\abs{q_3-\sqrt{E}\sin\left(2\sqrt{E}x\right)}\leq \frac{2\sin^2k}{\sin^2k-2\delta_n(k)}\left(6A+\frac{\sqrt{E}\delta_n(k)}{\sin^2k}\right).\nonumber
	\end{align}
	Using \eqref{Ssinrhoxtorho}, \eqref{Sprimecosrhox}, \eqref{C1} and \eqref{Ssquared1}, we can obtain
	\begin{align}
		\abs{q_4}\leq \frac{4\delta_n(k)^2}{\sin^2 k-2\delta_n(k)}\left(\sqrt{E}+12A\right).\nonumber
	\end{align}
	Therefore, applying \eqref{oneto16}, \eqref{finalsink} and \eqref{finalsink^2}, we have
	\begin{align}
			\abs{\left(\abs{\varphi(x,E)}^2\right)'}=\abs{\sum_{i=1}^4q_i}\leq \frac{10\sqrt{E}\delta_n(k)}{\sin^2k}+30A.\nonumber
	\end{align}
	By \eqref{eknakndeltank} and the definitions of $\delta_n(k),\delta(k)$, there holds
	\begin{align}
		\frac{10\sqrt{E}\delta_n(k)}{\sin^2k}+30A\leq {\frac{{\delta}(k)}{2}}.\nonumber
	\end{align}
	Then by \eqref{etaprime}, \eqref{omegavarphiestimate} and \eqref{varphixlargethan1-16}, one has
	\begin{align}
		\abs{\eta^{\prime\prime}(x,E)}
		=\abs{\frac{\omega^k_n\left(\abs{\varphi(x,E)}^2\right)'}{2\abs{\varphi(x,E)}^4}}
		\leq \frac{1}{2}(n\pi+{\delta}(k))\frac{{{\delta}(k)}}{1-\frac{16}{\sin^2k}\delta_n(k)}.\label{finaletaprimeprime}
	\end{align}
	Since $n>L(k)$, one can obtain 
	\begin{align}
		n\pi>\frac{2\times 10^3\pi A}{\sin^3k\sin\frac{99k}{100}\sin\frac{\pi+99k}{100}}\geq2\delta(k),\nonumber
	\end{align}
and 
\begin{align}
	\frac{\delta_n(k)}{\sin^2k}= \frac{Ae^{3}}{n\left(\sin^3 k\sin\frac{99k}{100}\sin\frac{\pi+99k}{100}\right)}\leq \frac{1}{100}.\nonumber
\end{align}
Then by \eqref{finaletaprimeprime} we obtain \eqref{etaprimeprime}.
\end{proof}

\section{oscillatory integral estimates coupled with periodic term}
In this section, we will introduce some important integral estimates which will be used in the final proof.

In the proof of \cite[Lemma 3.1]{L19stark},	the author has actually proved the following lemma without the quantitative description. For completeness, we will provide with a proof in the appendix.

\begin{lemma}\label{L1}
	Let $\beta\in (0,1]$. 
	Suppose that $\theta(t)$ satisfies
	\begin{align}\label{abstheta'xminu1lessthanc}
		\abs{\theta'(t)-1}\leq \frac{C}{t^\beta}\ \ a.e. \ t\in (0,\infty),
	\end{align}
	then for any $x\geq x_0>s(\beta,C)=(100C+10^4)^{\frac{1}{\beta}}$, one has
	\begin{align}\label{sinthetatovert}
		\abs{\int_{x_0}^x\frac{\sin\theta(t)}{t}dt}\leq \frac{{30C}{\beta^{-1}}+10\pi}{x_0^\beta},
	\end{align}
	and
	\begin{align}\label{costhetatovert}
		\abs{	\int_{x_0}^x\frac{\cos\theta(t)}{t}dt}\leq \frac{{30C}{\beta^{-1}}+10\pi}{x_0^\beta}.
	\end{align}
\end{lemma}

Using Lemma \ref{L1}, we are able to obtain the estimates for more general oscillatory integrals.
\begin{lemma}\label{lem25ts}
	Let $\beta\in (0,1]$ and $a\neq0$ be constants. 
	Suppose that $\theta(t)$ satisfies
	\begin{align}\label{abstheta'xminualessthanc}
		\abs{\theta'(t)-a}\leq \frac{C}{t^\beta}\ \ a.e.\ t\in (0,\infty).	\end{align}
	Then for any $x\geq x_0>s(\abs{a},\beta,C)=
	(100C\abs{a}^{\beta-1}+10^4)^{\frac{1}{\beta}}\abs{a}^{-1}$, one has
	\begin{align}\label{sinthetatoverta}
		\abs{\int_{x_0}^x\frac{\sin\theta(t)}{t}dt}\leq \frac{r(\abs{a},\beta,C)}{\abs{a}^{\beta}x_0^\beta},
	\end{align}
	and
	\begin{align}\label{costhetatoverta}
		\abs{	\int_{x_0}^x\frac{\cos\theta(t)}{t}dt}\leq \frac{r(\abs{a},\beta,C) }{\abs{a}^{\beta}x_0^\beta},
	\end{align}
	where 
	$r(\abs{a},\beta,C)=30C\beta^{-1}\abs{a}^{-1+\beta}+10\pi.$
\end{lemma}
\begin{proof}We only give the proof of \eqref{sinthetatoverta}, the proof of \eqref{costhetatoverta} is similar and omitted. Without loss of generality, we assume $a>0$.
	Change variables with $t=\frac{1}{a}s$.	Define
	\begin{align}
		\tilde{\theta}(s)=\theta(t)=\theta\left(\frac{1}{a}s\right).\nonumber
	\end{align}
	Then by \eqref{abstheta'xminualessthanc}, one has
	\begin{align}
		\abs{\tilde{\theta}'(s)-1}\leq \frac{\frac{C}{a^{1-\beta}}}{s^\beta} \ \ a.e. \ s\in (0,\infty).\nonumber
	\end{align}
	Since $x_0>s(\abs{a},\beta,C)$, we can obtain 
	\begin{align}
		ax_0>s(\beta,\frac{C}{a^{1-\beta}})=\left(100Ca^{\beta-1}+10^4\right)^{\frac{1}{\beta}}.\nonumber
	\end{align}
	Thus, by \eqref{sinthetatovert}, we have
	\begin{align}
		\abs{\int_{x_0}^x\frac{\sin\theta(t)}{t}dt}=	\abs{\int_{ax_0}^{ax}\frac{\sin\tilde{\theta}(s)}{s}ds}\leq \frac{30C\beta^{-1}a^{-1+\beta}+10\pi }{(ax_0)^\beta}.\nonumber
	\end{align}
	Here we finish the proof.
\end{proof}


\begin{lemma}\label{L2}Let $a\in \mathbb{R}\setminus 2\pi\mathbb{Z}$ and $\alpha=\min\{\abs{a-2\pi n}: n\in\mathbb{Z}\}$. Let $s(\alpha,\beta,C)$ and  $r(\alpha,\beta,C)$ be given as in Lemma \ref{lem25ts}.
Assume that $\gamma(\cdot)\in AC_{loc}(0,\infty),\Gamma(\cdot)\in L^2_{loc}(0,\infty)$. Assume more that $\gamma(t)\mod 2\pi,\Gamma(t)$ are $1$-periodic.  Suppose that $\theta(t)$ satisfies
	\begin{align}
		\abs{\theta'(t)-a-\gamma'(t)}\leq \frac{1}{t^{\frac{2}{3}}}\ \ a.e.\ t\in(0,\infty),\nonumber
	\end{align}
 then for any $x\geq x_0>r(\alpha):=s(
\alpha,\frac{2}{3},1)+4r(
\alpha,\frac{2}{3},1)\left({\alpha^{-\frac{4}{3}}}+3\right)^{\frac{1}{2}}$, one has
	\begin{align}\label{Gammasintheta}
		\abs{\int_{x_0}^x\frac{\Gamma(t)\sin\theta(t)}{t}dt}\leq \frac{\left\lVert \Gamma\right\rVert_{2}r(\alpha)}{x_0^{\frac{2}{3}}},
	\end{align}
	where $\left\lVert \Gamma\right\rVert_{2}=\left(\int_0^1\Gamma(t)^2dt\right)^{\frac{1}{2}}$.
\end{lemma}

\begin{proof}
	Denote by 
	\begin{align}
		\tilde{\theta}(t)=\theta(t)-\gamma(t).\nonumber
	\end{align}
	Then one has
	\begin{align}\label{thetatilde}
		\abs{\tilde{\theta}'(t)-a}\leq \frac{1}{{t}^{\frac{2}{3}}} \ \ a.e.\ t\in (0,\infty),
	\end{align}
	and
	\begin{align}
		\sin\theta(t)=\cos\gamma(t)\sin\tilde{\theta}(t)+\sin\gamma(t)\cos\tilde{\theta}(t).\nonumber
	\end{align}
	Hence, we only need to prove
	\begin{align}\label{Gammasinthetaminus}
		\abs{\int_{x_0}^{x}\frac{\Gamma(t)\cos\gamma(t)\sin\tilde{\theta}(t)}{{t}}dt}\leq \frac{2\left\lVert \Gamma\right\rVert_{2}r(
\alpha,\frac{2}{3},1)\left({\alpha^{-\frac{4}{3}}}+3\right)^{\frac{1}{2}}}{x_0^{\frac{2}{3}}},
	\end{align}
	and
	\begin{align}\label{Gammacosthetaminus}
		\abs{\int_{x_0}^{x}\frac{\Gamma(t)\sin\gamma(t)\cos\tilde{\theta}(t)}{{t}}dt}\leq \frac{2\left\lVert \Gamma\right\rVert_{2}r(
\alpha,\frac{2}{3},1)\left({\alpha^{-\frac{4}{3}}}+3\right)^{\frac{1}{2}}}{x_0^{\frac{2}{3}}}.
	\end{align}
	To avoid repetition, we only prove \eqref{Gammasinthetaminus}. We mention that $\Gamma(t)\cos\gamma(t)$ is still $1$-periodic. 
	Consider the Fourier expansion of $\Gamma(t)\cos\gamma(t)$,
	\begin{align}
		\Gamma(t)\cos\gamma(t)=\frac{a_0}{2}+\sum_{n=1}^\infty(a_n\cos\left(2\pi nt\right)+b_n\sin\left(2\pi nt\right)),\nonumber
	\end{align}
	where $a_0,a_n,b_n(n\geq 1)$ are the Fourier coefficients of $\Gamma(t)\cos\gamma(t)$.
	Then
	\begin{align}
		&\abs{\int_{x_0}^{x}\frac{\Gamma(t)\cos\gamma(t)\sin\tilde{\theta}(t)}{{t}}dt}\leq \left(\frac{a_0^2}{2}+\sum_{n=1}^\infty (a_n^2+ b_n^2)\right)^{\frac{1}{2}}\cdot
		\Bigg(\frac{1}{2}\left(\int_{x_0}^x\frac{\sin\tilde{\theta}(t)}{t}dt\right)^2\nonumber\\
		&+\sum_{n=1}^\infty  \left(\int_{x_0}^x\frac{\sin\tilde{\theta}(t)\cos\left(2\pi nt\right)}{t}dt\right)^2+\sum_{n=1}^\infty \left(\int_{x_0}^x\frac{\sin\tilde{\theta}(t)\sin\left(2\pi nt\right)}{t}dt\right)^2   \Bigg)^{\frac{1}{2}}.\nonumber
	\end{align}
	By Parseval's identity, one has
$2\int_0^1\Gamma(t)^2\cos^2\gamma(t)dt=\frac{a_0^2}{2}+\sum_{n=1}^\infty(a_n^2+b_n^2).$
	Denote by
	\begin{align}
		M=&\frac{1}{2}\left(\int_{x_0}^x\frac{\sin\tilde{\theta}(t)}{t}dt\right)^2+\sum_{n=1}^\infty  \left(\int_{x_0}^x\frac{\sin\tilde{\theta}(t)\cos\left(2\pi nt\right)}{t}dt\right)^2\nonumber\\
		&+\sum_{n=1}^\infty \left(\int_{x_0}^x\frac{\sin\tilde{\theta}(t)\sin\left(2\pi nt\right)}{t}dt\right)^2.\nonumber
	\end{align}
	To prove\eqref{Gammasinthetaminus}, we only need to show 
	\begin{align}
		M\leq \frac{2r(
\alpha,\frac{2}{3},1)^2\left({\alpha^{-\frac{4}{3}}}+3\right)}{x_0^{\frac{4}{3}}}.\label{M}
	\end{align}
	Note that for any $n\in\mathbb{N}$, we have 
	\begin{align}
		x_0>s\left(\alpha,\frac{2}{3},1\right)\geq s\left(\abs{a\pm2\pi n},\frac{2}{3},1\right),\label{salpha}
	\end{align}
	and
	\begin{align}\label{alpha}
		r\left(\alpha,\frac{2}{3},1\right)\geq r\left(\abs{a\pm2\pi n},\frac{2}{3},1\right).
	\end{align}

	By \eqref{sinthetatoverta}, \eqref{thetatilde}, \eqref{salpha} and \eqref{alpha}, one has
	\begin{align}\label{finallemma0}
		\abs{\int_{x_0}^x\frac{\sin\tilde{\theta}(t)}{t}dt}\leq \frac{	r\left(\alpha,\frac{2}{3},1\right)}{\abs{a}^{\frac{2}{3}}x_0^{\frac{2}{3}}}.
	\end{align}
	Since $\sin\tilde{\theta}(t)\cos\left(2\pi nt\right)=\frac{\sin(\tilde{\theta}(t)+2\pi nt)+\sin(\tilde{\theta}(t)-2\pi nt)}{2}$,
	by \eqref{sinthetatoverta}, \eqref{thetatilde}, \eqref{salpha} and \eqref{alpha} one can obtain
	\begin{align}
		\abs{\int_{x_0}^x\frac{\sin\tilde{\theta}(t)\cos\left(2\pi nt\right)}{t}dt}\leq \frac{	r\left(\alpha,\frac{2}{3},1\right)}{2x_0^{\frac{2}{3}}}\left(\frac{1}{\abs{a+2\pi n}^{\frac{2}{3}}}+\frac{1}{\abs{a-2\pi n}^{\frac{2}{3}}}\right).\label{finallemma}
	\end{align}
	Similarly, one has
	\begin{align}
		\abs{\int_{x_0}^x\frac{\sin\tilde{\theta}(t)\sin\left(2\pi nt\right)}{t}dt}\leq \frac{	r\left(\alpha,\frac{2}{3},1\right)}{2x_0^{\frac{2}{3}}}\left(\frac{1}{\abs{a+2\pi n}^{\frac{2}{3}}}+\frac{1}{\abs{a-2\pi n}^{\frac{2}{3}}}\right).\label{final2lemma}
	\end{align}
	Since  $\alpha=\min\{\abs{a-2\pi n}:n\in\mathbb{Z}\}>0$, one can obtain
	\begin{align}
		\sum_{n=-\infty}^\infty\frac{1}{\abs{a+2\pi n}^{\frac{4}{3}}}\leq \frac{2}{\alpha^{\frac{4}{3}}}+6.\label{sumalphabeta}
	\end{align}
	Therefore, by \eqref{finallemma0}-\eqref{final2lemma} we can obtain 
	\begin{align}
		M\leq \frac{	r\left(\alpha,\frac{2}{3},1\right)^2}{{x_0}^{\frac{4}{3}}}\left(\frac{1}{2\abs{a}^{\frac{4}{3}}}+\sum_{n=1}^\infty\frac{1}{\abs{a+2\pi n}^{\frac{4}{3}}}+\sum_{n=1}^\infty\frac{1}{\abs{a-2\pi n}^{\frac{4}{3}}}\right)
	\end{align}
	Then by \eqref{sumalphabeta} one can obtain \eqref{M}.
	
	Here we finish the proof.
\end{proof}

\section{resonant integral estimates}
Before the final construction, we need some (non)resonant oscillatory integral estimates. In the following, for any $E_i=E^{k(E_i)}_{n(E_i)}\in I^\circ_{n(E_i)},E_j=E^{k(E_j)}_{n(E_j)}\in I^\circ_{n(E_j)}$, we still denote by $k_i=k(E_i),k_j=k(E_j),n_i=n(E_i),n_j=n(E_j)$.
\begin{lemma}[Non-resonance]
	Let $E_i, E_j\in \sigma_{ess}^\circ$ with $k_j\neq k_i, k_j\neq \tilde{k}_i$. Suppose that we have
	\begin{align}
		\abs{\frac{V(x)}{\eta'(x,E_l)}}\leq \frac{1}{8x^{\frac{2}{3}}}\label{lemma41tofrac23}
	\end{align}
for any $x\geq x_0\geq r(\alpha_1)$ and $l=i,j$.
Then one has
\begin{align}
	\label{cosnonransonanceij}\abs{	\int_{x_0}^x\frac{\sin2\theta(t,E_i)\sin2\theta(t,E_j)}{\eta'(t,E_i)t}dt}\leq\frac{\left\lVert \frac{1}{\eta'(\cdot,E_i)}\right\rVert_{2}r(\alpha_1)}{x_0^{\frac{2}{3}}},
\end{align}
where $\alpha_1=\min\{2\abs{k_i-k_j},2\abs{k_i+k_j-\pi}\}>0$, $r(\alpha)$ is defined as in Lemma \ref{L2}.

Moreover, if we have \eqref{lemma41tofrac23} for $l=i$ and any $x\geq x_0\geq r(\alpha_2)$, and $k_i\neq\frac{\pi}{2}$, then one has
	\begin{align}
	\label{cosnonranaonance4theta}\abs{	\int_{x_0}^x\frac{\cos4\theta(t,E_i)}{\eta'(t,E_i)t}dt}\leq\frac{\left\lVert \frac{1}{\eta'(\cdot,E_i)}\right\rVert_{2}r(\alpha_2)}{x_0^{\frac{2}{3}}},
	\end{align}
where $\alpha_2=\min\{{4k_i},\abs{4k_i-2\pi},4\pi-4k_i\}>0$.
\end{lemma}
\begin{proof}We only give the proof of \eqref{cosnonransonanceij}, the proof of \eqref{cosnonranaonance4theta} is similar and omitted.
We only need to show
	\begin{align}
	\abs{	\int_{x_0}^x\frac{\cos(2\theta(t,E_i)\pm2\theta(t,E_j))}{\eta'(t,E_i)t}dt}\leq\frac{	\left\lVert \frac{1}{\eta'(\cdot,E_i)}\right\rVert_{2}r(\alpha_1)}{x_0^{\frac{2}{3}}}.\label{middle107intcos}
\end{align}
	By \eqref{etakegamma}, \eqref{prufertheta24} and \eqref{lemma41tofrac23},  one has
	\begin{align}
		\abs{\left(2(\theta(t,E_i)\pm \theta(t,E_j))+\frac{\pi}{2}\right)'-2(k_i\pm k_j)-2(\gamma'(t,E_i)\pm\gamma'(t,E_j))}\leq \frac{1}{t^{\frac{2}{3}}}.\nonumber
	\end{align}
Since $\alpha_{\pm}:=\min\{\abs{2k_i\pm 2k_j-2\pi n}:n\in\mathbb{Z}\}\geq\alpha_1$, one has $r(\alpha_{\pm})\leq r(\alpha_1)$.
By \eqref{Gammasintheta} one can obtain \eqref{middle107intcos}.
	
\end{proof}
\begin{lemma}[Resonances from same band]\label{lemmavarepsilon} For any $E_i,E_j\in \sigma_{ess}^\circ$ with  $k_j=\tilde{k}_i$ and $n_j=n_i$ (we do not exclude that $E_i=E_j$), there exists $\varepsilon_i\left(=\varepsilon_j\right)\in (0,1)$ such that if for any $x\geq x_0$ one has
	\begin{align}
		\abs{\frac{V(x)}{\eta'(x,E_l)}}\leq \frac{\varepsilon_i}{2},\  l=i,j,\label{voveretaplusvoveretatilde}
	\end{align} 
then we have
\begin{align}\label{oneminuscos2theta>bkn}
		\int_{x_0}^{x}\frac{1-\cos(2\theta(t,E_i)+2\theta(t,E_j))}{\abs{\eta'(t,E_i)}t}dt>\frac{\varepsilon_i}{A_i}\ln \frac{x}{x_0+1},
\end{align}and
\begin{align}\label{oneminuscos2theta>bkn2}
	\int_{x_0}^{x}\frac{1-\cos(2\theta(t,E_i)+2\theta(t,E_j))}{\abs{\eta'(t,E_j)}t}dt>\frac{\varepsilon_i}{A_i}\ln \frac{x}{x_0+1},
\end{align}
where $A_i=\max\left\{\max_{x\in [0,1]}\abs{\eta'(x,E_i)},\max_{x\in [0,1]}\abs{\eta'(x,E_j)}\right\}.$
\end{lemma}
\begin{proof} We only need to prove \eqref{oneminuscos2theta>bkn}, then interchange $i$ and $j$, and let $\varepsilon_i=\varepsilon_j$ be small enough, one can obtain  \eqref{oneminuscos2theta>bkn2}.
	By \eqref{etakegamma}, one has
	\begin{align}
		\eta'(x,E_i)=k_i+\gamma'(x,E_i),\
		\eta'(x,E_j)=\tilde{k}_i+\gamma'(x,E_j),\nonumber
\end{align}
where $\gamma(x,E_i)\mod 2\pi$ and  $\gamma(x,E_j)\mod 2\pi$ are  $1$-periodic functions in $x$. 
Then by \eqref{prufertheta24}, one has
\begin{align}
	2\theta'(x,E_i)+2\theta'(x,E_j)=2\pi +\gamma'(x)+O\left(	\abs{\frac{V(x)}{\eta'(x,E_i)}}+\abs{\frac{V(x)}{\eta'(x,E_j)}}\right)\nonumber
\end{align}
as $x\to\infty$,
where $\gamma(x)\mod4\pi$ is a $1$-periodic function in $x$. From \eqref{voveretaplusvoveretatilde} we have that for small enough positive $\varepsilon_i$ and any $t>x_0$, there exists $x_t\in (t,t+1)$ such that $2\theta(x_t,E_i)+2\theta(x_t,E_j)=\pi\mod2\pi$. Moreover, since  $2\theta'(x,E_i)+2\theta'(x,E_j)=O(1)$ as $x\to\infty$,
we know that for small enough positive $\varepsilon_i$ we have that for any $x>x_0$,
\begin{align}
	\int_{x}^{x+1}{\big(1-\cos(2\theta(t,E_i)+2\theta(t,E_j))\big)}dt>\varepsilon_i.\label{middlelargethanvarepsion}
\end{align}
Suppose that $x_0+M\leq x< x_0+M+1$, then
\begin{align}
\int_{x_0}^{x}\frac{1-\cos(2\theta(t,E_i)+2\theta(t,E_j))}{\abs{\eta'(t,E_i)}t}dt=&\sum_{l=0}^{M-1}\int_{x_0+l}^{x_0+l+1}\frac{1-\cos(2\theta(t,E_i)+2\theta(t,E_j))}{\abs{\eta'(t,E_i)}t}dt\nonumber\\
&+\int_{x_0+M}^x\frac{1-\cos(2\theta(t,E_i)+2\theta(t,E_j))}{\abs{\eta'(t,E_i)}t}dt\nonumber\\
\geq &\sum_{l=0}^{M-1}\int_{x_0+l}^{x_0+l+1}\frac{1-\cos(2\theta(t,E_i)+2\theta(t,E_j))}{A_i(x_0+l+1)}dt
\end{align}
Then by \eqref{middlelargethanvarepsion} one can obtain 
\begin{align}
	\int_{x_0}^{x}\frac{1-\cos(2\theta(t,E_i)+2\theta(t,E_j))}{\abs{\eta'(t,E_i)}t}dt>\frac{\varepsilon_i}{A_i}\sum_{l=0}^{M-1}\frac{1}{x_0+l+1}.\nonumber
\end{align}
Hence, one obtains \eqref{oneminuscos2theta>bkn}.

\end{proof}

\begin{lemma}[Resonances from different bands]\label{lemmafinal5.6}
	Assume that for some $E_i,E_j\in \sigma_{ess}^\circ$, one has $k_i=k_j$ or $k_i=\tilde{k}_j$. Suppose that $n_i>L(k_i)$,  $n_j>L(k_j)(=L(k_i))$ and $n_i\neq n_j$, then if for any $x>x_0$ one has
	\begin{align}
			\abs{\frac{V(x)}{\eta'(x,E_l)}}\leq \frac{\delta(k_i)}{2},\ l=i,j,\label{lemmatildedetalki}
			\end{align}
	 we have 
	 \begin{align}
	 	\label{con2thetaipm2thetaj}
	 	\abs{\int_{x_0}^{x}\frac{\sin2\theta(t,E_i)\sin2\theta(t,E_j)}{\eta'(t,E_i)t}dt}<\frac{10+10{\delta}(k_i)}{\abs{n_i-n_j}n_i}\ln{x}+\frac{O(1)}{x_0},
	 \end{align}
 and 
 \begin{align}
 	\abs{\int_{x_0}^{x}\frac{\cos4\theta(t,E_i)}{\eta'(t,E_i)t}dt}<\frac{10+10{\delta}(k_i)}{n_i^2}\ln x+\frac{O(1)}{x_0}.\label{cos4thetat}
 \end{align}
\end{lemma}
\begin{proof}
	We only give the proof of \eqref{con2thetaipm2thetaj}, the proof of \eqref{cos4thetat} is similar and omitted. To prove \eqref{con2thetaipm2thetaj}, we only need to show
\begin{align}
	\abs{\int_{x_0}^{x}\frac{\cos(2\theta(t,E_i)\pm2\theta(t,E_j))}{\eta'(t,E_i)t}dt}\leq\frac{10+10{\delta}(k_i)}{\abs{n_i-n_j}n_i}\ln{x}{}+\frac{O(1)}{x_0}.\label{middlecos2thetapmij}
\end{align}
To avoid repetition, we assume that $n_i$ and $n_j$ are both odd. Then by \eqref{prufertheta24} and \eqref{omegavarphiestimate} one can obtain that for some $g_1(t), g_2(t)$,
\begin{align}
	2\theta(t,E_i)=2n_i\pi t+g_1(t),	\ 2\theta(t,E_j)=2n_j\pi t+g_2(t),\label{middleg1g2}
\end{align} 
where
\begin{align}
	\abs{g'_1(t)}\leq 2\pi +4{\delta}(k_i),\ \ \abs{g_2'(t)}\leq 2\pi +4{\delta}(k_i).\nonumber
\end{align}
By the definitions of $L(k_i)$ and ${\delta}(k_i)$, one has ${\delta}(k_i)\leq\frac{L(k_i)\pi-\pi}{2} \leq\frac{n_i\pi-\pi}{2}$. Hence, by \eqref{omegavarphiestimate} one obtains
\begin{align}
	\frac{1}{\abs{\eta'(t,E_i)}}\leq\frac{2}{(n_i-1)\pi}\leq \frac{4}{n_i\pi}.\nonumber
\end{align}
Then by
\begin{align}
	\abs{\left(\frac{\cos(g_1(t)\pm g_2(t))}{\eta'(t,E_i)t}\right)'}\leq \frac{\abs{g_1'(t)\pm g_2'(t)}}{\eta'(t,E_i)t}+\frac{\abs{\eta^{\prime\prime}(t,E_i)}}{\eta'(t,E_i)^2t}+\frac{1}{\eta'(t,E_i)t^2}\nonumber
\end{align}
and  \eqref{etaprimeprime}, one has
\begin{align}
	\abs{\left(\frac{\cos(g_1(t)\pm g_2(t))}{\eta'(t,E_i)t}\right)'}\leq \frac{16 +16{\delta}(k_i)}{n_i t}+\frac{O(1)}{t^2}.\nonumber
\end{align}
Therefore,
\begin{align}
	\abs{\int_{x_0}^{x}\frac{{\cos(2\pi t(n_i\pm n_j))\cos(g_1(t)\pm g_2(t))}}{\eta'(t,E_i)t}dt}
	=&\Bigg|\frac{{\sin(2\pi t(n_i\pm n_j))\cos(g_1(t)\pm g_2(t))}}{2\pi (n_i\pm n_j)\eta'(t,E_i)t}\Big|^{x}_{x_0}\nonumber\\
	-\int_{x_0}^{x}\frac{\sin(2\pi t(n_i\pm n_j))}{2\pi (n_i\pm n_j)}&\left(\frac{\cos(g_1(t)\pm g_2(t))}{\eta'(t,E_i)t}\right)'dt\Bigg|\nonumber\\
	\leq &\frac{5+5{\delta}(k_i)}{n_i\abs{n_i-n_j} }\ln x+\frac{O(1)}{x_0}.\label{middlecosg1g2eta}
\end{align}
Similarly, one has
\begin{align}
	\abs{\int_{x_0}^{x}\frac{{\sin(2\pi t(n_i\pm n_j))\sin(g_1(t)\pm g_2(t))}}{\eta'(t,E_i)t}dt}\leq \frac{5+5{\delta}(k_i)}{n_i\abs{n_i-n_j} }\ln x+\frac{O(1)}{x_0}.\label{middlesing1g2eta}
\end{align}
Hence, by \eqref{middleg1g2}, \eqref{middlecosg1g2eta} and \eqref{middlesing1g2eta}  we have \eqref{middlecos2thetapmij}.

\end{proof}

\section{construction of the perturbations}

Recall that $S=\{E_j\}_{j=1}^K=\{E^{k_j}_{n_j}\}_{j=1}^K$ with $K\leq \infty$, and $S=S_1\cup S_2\cup S_3$ satisfies the assumption \textbf{A1}.
For any $j\leq K$, denote by 
\begin{align}
	A_j=\max\left\{\max_{x\in [0,1]}\abs{\eta'(x,E^{k_j}_{n_j})},\max_{x\in [0,1]}\abs{\eta'(x,E^{\tilde{k}_j}_{n_j})}\right\},\label{deAj}\\
	{B_j=\min\left\{\min_{x\in [0,1]}\abs{\eta'(x,E^{k_j}_{n_j})},\min_{x\in [0,1]}\abs{\eta'(x,E^{\tilde{k}_j}_{n_j})}\right\}.\label{deBj}}
\end{align} 
 Define $\{C_j\}_{j=1}^K\subset\mathbb{R}$ by
\begin{align}
	C_j=\begin{cases}
		\frac{400A_j}{\varepsilon_j},&  E_j\in S_1\ \mathrm{with}\ k_j=\frac{\pi}{2}\ \mathrm{\ or} \ E_j\in S_2,\\
		{400A_j},& E_j\in S_1\ \mathrm{with}\ k_j\neq\frac{\pi}{2}\ \mathrm{\ or} \ E_j\in S_3,\nonumber
	\end{cases}
\end{align}
where $\varepsilon_j$ were given in Lemma \ref{lemmavarepsilon}.
We also need to take the problem that the eigenvalues will ``asymptotically resonate" into consideration when constructing the potentials. For instance for $K=\infty$,  $k_i\neq k_j$, $\abs{k_i-k_j}\to 0$, or  $k_i+k_j\neq\pi$, $\abs{k_i+k_j-\pi}\to0$ for some $i,j\to\infty$.
Define $\{\alpha_j\}_{j=1}^K\subset\mathbb{R}$ by
\begin{align}
	\alpha_j=\min\{d_j,e_j,f_j\},\label{definitionofDj}
\end{align}
where
\begin{align}	d_j=&\min\{4k_i,4\pi-4k_i:1\leq i\leq j,\ k_i\neq \frac{\pi}{2}\},\nonumber\\
	e_j=&\min\{2\abs{k_i-k_l}:1\leq i<l\leq j, k_i\neq k_l\},\nonumber\\
	f_j=\min&\{2\abs{k_i+k_l-\pi}:1\leq i\leq l\leq j,k_i+k_l\neq\pi\}.\nonumber
\end{align}
Let $r(\alpha)$ be given as in Lemma \ref{L2}, then on has
\begin{align}
	r(\alpha_1)\leq r(\alpha_2)\leq r(\alpha_3)\leq\cdots.\label{ralphabetacdecrease}
\end{align}

Let $h(x)$ be the function in Theorem \ref{theoremmain224} and for any $j\leq K$,
\begin{align}
	x_j:=\inf\{x:{h(t)}>10\sum_{i=1}^{j}C_i,\ {\rm{for\ any\ }} t>x\}.\label{definitionofxj}
\end{align}
Let $\{T_j\}_{j=1}^K\subset\R$ be any set satisfying 
\begin{align}
	T_j>\max\Bigg\{&x_j+T_{j-1},\left(8\left(1+\sum_{i=1}^j\frac{1}{B_i}\right)\sum_{l=1}^jC_l\right)^3, \sum_{i=1}^j\left(\frac{2}{B_i\varepsilon_i}+\frac{2}{B_i{\delta}(k_i)}\right)\sum_{l=1}^jC_l,\nonumber\\
	&
	r(\alpha_j)+\left(10^jr(\alpha_j)\sum_{i=1 }^jC_i\right)^{\frac{3}{2}},10^j\left(n_j+\sum_{l=1}^jC_l\right)\Bigg\},\label{detj}
\end{align}
here $T_0=0$.
The requirements on the largeness of $T_j$ are to ensure that the potentials $V$ defined later satisfy the smallness conditions in the previous lemmas, and some properties of convergence, we will give explicit explanations in the final proof.

 Now we introduce the construction for the perturbations. Instead of piecewise constructing perturbations like \cite{JL19,LO17,L19stark,vishwamdirac,kangperiodicdirac} and then  gluing each pieces together, our construction is more efficient in the aspect of the decreasing of eigensolutions. In our paper, the eigensolutions decrease continuously, not piecewise.
 Once we take the eigenvalue $E_j$ into consideration in $j$-th step, the eigensolution starts decreasing. This ensures that we are able to deal with some cases of resonance.  
 The construction is based on solving a class of differential equations.

Solve the following differential equations on $[0,\infty)$
 \begin{align}
 	\Theta_j'(x)=\eta'(x,E_j)-\frac{\sum_{j=1}^K(-1)^{n_j}C_j\sin2\Theta_j(x)\chi_{[T_j,\infty)}(x)}{(1+x)\eta'(x,E_j)}\sin^2\Theta_j(x),\label{Thetaprime}
 \end{align}
with the initial conditions $\Theta_j(0)=\Theta_j$, $j\leq K$, where $\chi(x)$ is the characteristic function. Then one can obtain solutions $\Theta_j(x)$ on $[0,\infty)$, $j\leq K$.
 Comparing \eqref{prufertheta24} and \eqref{Thetaprime}, we know that if we let $\Theta_j=\theta(0,E_j)$ for any $j\leq K$, and
\begin{align}
	V(x)=\frac{\sum_{j=1}^K(-1)^{n_j}C_j\sin2\Theta_j(x)\chi_{[T_j,\infty)}(x)}{1+x},\ x\in [0,\infty),\label{definitionofperturbation}
\end{align}
we can obtain $\Theta_j(x)=\theta(x,E_j)$ and 
\begin{align}
		V(x)=\frac{\sum_{j=1}^K(-1)^{n_j}C_j\sin2\theta(x,E_j)\chi_{[T_j,\infty)}(x)}{1+x},\ x\in [0,\infty).\label{definitionofpotentials}
\end{align}

We mention that
we have $K<\infty$ in the construction of Theorem \ref{theoremmain124} and we
instead have $K=\infty$ in the construction of Theorem \ref{theoremmain224}. We only need to prove the case $K=\infty$.

\begin{proof}[\textbf{Proof of Theorem \ref{theoremmain224}}]
Let the potential be defined by \eqref{definitionofpotentials} with $K=\infty$.  For any $\tilde{x}>0$ (assume $T_{m'}\leq \tilde{x}< T_{m'+1}$), from \eqref{detj} we obtain $T_{m'}>x_{m'}$, by \eqref{definitionofxj}, one has
\begin{align}
	\abs{V(\tilde{x})}\leq \frac{\sum_{l=1}^{m'}C_l}{1+\tilde{x}}\leq \frac{h(\tilde{x})}{1+\tilde{x}}.\nonumber
\end{align}

Next,
we show that for any $m\in \mathbb{N}$, $R(\cdot,E_m)\in L^2(0,\infty).$ 
For any $x\geq T_m,$
suppose that $T_M\leq x<T_{M+1}$. 
By \eqref{etaprime}, \eqref{pruferR24}, \eqref{definitionofpotentials} and Remark \ref{omegaknpm}, one has
\begin{align}
	\ln R(x,E_m)-&\ln R(T_m,E_m)\nonumber\\
	=&\sum_{j=1}^M(-1)^{n_j}\int_{T_m}^x\frac{C_j\sin2\theta(t,E_j)\sin2\theta(t,E_m)\chi_{[T_j,\infty)}(t)}{2\eta'(t,E_m)(1+t)}dt\nonumber\\
	=&-\int_{T_m}^x\frac{C_m(1-\cos4\theta(t,E_m))}{4\abs{\eta'(t,E_m)}(1+t)}dt\nonumber\\
	&+\sum_{\substack{j\neq m,j\leq M\\j\in I}}(-1)^{n_j}\int_{T_m}^x\frac{C_j\sin2\theta(t,E_j)\sin2\theta(t,E_m)\chi_{[T_j,\infty)}(t)}{2\eta'(t,E_m)(1+t)}dt\nonumber\\
	&+\sum_{\substack{j\leq M,j\notin I}}(-1)^{n_j}\int_{T_m}^x\frac{C_j\sin2\theta(t,E_j)\sin2\theta(t,E_m)\chi_{[T_j,\infty)}(t)}{2\eta'(t,E_m)(1+t)}dt\nonumber\\
	:=&P_1+\sum_{\substack{j\neq m,j\leq M\\j\in I}}(-1)^{n_j}P_2(j)+\sum_{\substack{j\leq M,j\notin I}}(-1)^{n_j}P_3(j),\label{lnrp1p2p3tmem}
\end{align}
where
\begin{align}
	I=\{j: k_j=k_m\ \mathrm{or}\ k_j=\tilde{k}_m\}.\nonumber
\end{align}

Firstly, let us consider $P_3(j)$. If $j<m$, for any $\tilde{x}\geq T_m$ (assume $T_m\leq T_{m'}\leq \tilde{x}<T_{m'+1}$),
from \eqref{detj}, there holds $T_{m'}^{\frac{{1}}{3}}>\sum_{i=1}^{m'}\frac{8}{B_i}\sum_{l=1}^{m'}C_l$. Then by \eqref{deBj} and \eqref{definitionofpotentials} one has
\begin{align}
	\abs{\frac{V(\tilde{x})}{\eta'(\tilde{x},E_j)}}\leq\frac{\sum_{l=1}^{m'}C_l}{B_j\tilde{x}}\leq \frac{\sum_{l=1}^{m'}C_l}{B_jT_{m'}^{\frac{1}{3}}\tilde{x}^{\frac{2}{3}}}\leq \frac{1}{8\tilde{x}^{\frac{2}{3}}},\nonumber
\end{align}
and
\begin{align}
		\abs{\frac{V(\tilde{x})}{\eta'(\tilde{x},E_m)}}\leq\frac{1}{8\tilde{x}^{\frac{2}{3}}},\nonumber
\end{align}
then one obtains \eqref{lemma41tofrac23} with $l=m,j$. From \eqref{detj} we obtain $T_m>r(\alpha_m)$, by \eqref{cosnonransonanceij}, \eqref{definitionofDj} and \eqref{ralphabetacdecrease} there holds
\begin{align}
	P_3(j)
	=\int_{T_m}^x\frac{C_j\sin2\theta(t,E_j)\sin2\theta(t,E_m)}{2\eta'(t,E_m)(1+t)}dt
	=O\left(\frac{C_jr(\alpha_m)}{T_m^{\frac{2}{3}}}\right).\nonumber
\end{align}
If $j>m$, then one similarly has
\begin{align}
	P_3(j)=&\int_{T_j}^x\frac{C_j\sin2\theta(t,E_j)\sin2\theta(t,E_m)}{2\eta'(t,E_m)(1+t)}dt
	=O\left(\frac{C_jr(\alpha_j)}{T_j^{\frac{2}{3}}}\right).\nonumber
\end{align}
From \eqref{detj} we have 
$
	T_j^{\frac{2}{3}}\geq 10^jr(\alpha_j)\sum_{i\leq j}C_i,
$
then one obtains
\begin{align}
	\sum_{\substack{j\leq M,j\notin I}}(-1)^{n_j}P_3(j)=O(1).\label{p3O1}
\end{align}

Let us now consider  $P_1$ and $P_2(j)$. 

If $E_m\in S_1$, then $I=\{m\}$. Namely, there is no eigenvalue which creates resonance with $E_m$, then 
\begin{align}
	P_2(j)=0.\label{P2j=0}
\end{align}
Furthermore,
if $k_m\neq \frac{\pi}{2}$, then $C_m=400A_m$, by \eqref{cosnonranaonance4theta}, \eqref{definitionofDj} and \eqref{ralphabetacdecrease}, one has
\begin{align}
	\int_{T_m}^x\frac{C_m\cos4\theta(t,E_m)}{4\abs{\eta'(t,E_m)}(1+t)}dt=O(1).\nonumber
\end{align}
Then by \eqref{deAj} we obtain
\begin{align}
	P_1=-100A_m\int_{T_m}^x\frac{1}{\abs{\eta'(t,E_m)}(1+t)}dt+O(1)
	\leq-100\ln {x}{}+O(1).\label{fip125731}
\end{align}
By \eqref{lnrp1p2p3tmem}-\eqref{fip125731} we can obtain $R(\cdot,E_m)\in L^2(0,\infty)$.

If $k_m=\frac{\pi}{2}$, then $C_m=\frac{400A_m}{\varepsilon_m}$. For any $\tilde{x}\geq T_m$ (assume $T_m\leq T_{m'}\leq \tilde{x}<T_{m'+1}$), since (from \eqref{detj}) $$T_{m'}>\sum_{i=1}^{m'}\left(\frac{2}{B_i\varepsilon_i}+\frac{2}{B_i{\delta}(k_i)}\right)\sum_{l=1}^{m'}C_l>\frac{2}{B_m\varepsilon_m}\sum_{l=1}^{m'}C_l,$$
by \eqref{deBj} we have
\begin{align}
	\abs{\frac{V(\tilde{x})}{\eta'(\tilde{x},E_m)}}\leq \frac{\sum_{l=1}^{m'}C_l}{B_m\tilde{x}}\leq\frac{\sum_{l=1}^{m'}C_l}{B_mT_{m'}}\leq\frac{\varepsilon_m}{2}.\label{middleconditionlemma5.5}
\end{align}
Hence we can obtain \eqref{voveretaplusvoveretatilde} with $l=m$. Then by \eqref{oneminuscos2theta>bkn} ($i=j=m$) one has
\begin{align}
	P_1=O(1)-\frac{100A_m}{\varepsilon_m}\int_{T_m}^x\frac{1-\cos4\theta(t,E_m)}{\abs{\eta'(t,E_m)}(1+t)}dt
	\leq -100\ln {x}+O(1).\nonumber
\end{align} 
Combining with \eqref{lnrp1p2p3tmem}, \eqref{p3O1} and \eqref{P2j=0} we can conclude $R(\cdot,E_m)\in L^2(0,\infty)$.

Therefore, if $E_m\in S_1$, one has proved the theorem.

Next, let us consider that $E_m\in S_2$. In this case
$k_m\neq \frac{\pi}{2}$, and there exists $E_{\tilde{m}}\in S$ such that $n_{\tilde{m}}=n_m$ and $k_{\tilde{m}}=\tilde{k}_m$, we mention that for such case  $C_{\tilde{m}}=C_m=\frac{400A_m}{\varepsilon_m}$, $I=\{m,\tilde{m}\}$. Without loss of generality, we assume $x>T_{\tilde{m}}>T_m$, then 
\begin{align}
	P_1+&\sum P_2(j)\nonumber\\
=&-\int_{T_m}^x\frac{C_m(1-\cos4\theta(t,E_m))}{4\abs{\eta'(t,E_m)}(1+t)}dt-\int_{T_{\tilde{m}}}^x\frac{C_m\sin2\theta(t,E_{\tilde{m}})\sin2\theta(t,E_m)}{2\abs{\eta'(t,E_m)}(1+t)}dt\nonumber\\
	=&O(1)-C_m\int_{T_{\tilde{m}}}^x\frac{1-\cos(2\theta(t,E_m)+2\theta(t,E_{\tilde{m}}))}{4\abs{\eta'(t,E_m)}(1+t)}dt+C_m\int_{T_m}^x\frac{\cos4\theta(t,E_m)}{4\abs{\eta'(t,E_m)}(1+t)}dt\nonumber\\
	&-C_m\int_{T_{\tilde{m}}}^x\frac{\cos(2\theta(t,E_m)-2\theta(t,E_{\tilde{m}}))}{4\abs{\eta'(t,E_m)}(1+t)}dt.\nonumber
\end{align}
Then by \eqref{etakegamma}, \eqref{prufertheta24}, \eqref{Gammasintheta} and \eqref{cosnonranaonance4theta}, we have
\begin{align}
	C_m\int_{T_{{m}}}^x\frac{\cos4\theta(t,E_m)}{4\abs{\eta'(t,E_m)}(1+t)}dt-C_m\int_{T_{\tilde{m}}}^x\frac{\cos(2\theta(t,E_m)-2\theta(t,E_{\tilde{m}}))}{4\abs{\eta'(t,E_m)}(1+t)}dt=O(1).\nonumber
\end{align}
Similar to \eqref{middleconditionlemma5.5}, we can obtained \eqref{voveretaplusvoveretatilde} with $l=m,\tilde{m}$ for any $\tilde{x}\geq T_{\tilde{m}}$. Then by Lemma \ref{lemmavarepsilon} one has
\begin{align}
	-C_m\int_{T_{\tilde{m}}}^x\frac{1-\cos(2\theta(t,E_m)+2\theta(t,E_{\tilde{m}}))}{4\abs{\eta'(t,E_m)}(1+t)}dt\leq -100\ln {x} +O(1).\nonumber
\end{align}
Therefore, we can obtain 
\begin{align}
	P_1+\sum P_2(j)\leq -100\ln x +O(1).\label{middle-100}
\end{align}
Then one can obtain $R(\cdot,E_m)\in L^2(0,\infty)$ by \eqref{lnrp1p2p3tmem} and \eqref{p3O1}.

Finally, we assume that $E_m\in S_3$.  One concludes that
\begin{align}
	P_1+\sum_{\substack{j\neq m,j\leq M\\j\in I}}(-1)^{n_j}P_2(j)=&-\int_{T_m}^x\frac{C_m(1-\cos4\theta(t,E_m))}{4\abs{\eta'(t,E_m)}(1+t)}dt\nonumber\\
	&+\sum_{\substack{j<m,j\in I}}(-1)^{n_j}\int_{T_m}^x\frac{C_j\sin2\theta(t,E_j)\sin2\theta(t,E_m)}{2\eta'(t,E_m)(1+t)}dt\nonumber\\
	&+\sum_{\substack{m<j\leq M\\j\in I}}(-1)^{n_j}\int_{T_j}^x\frac{C_j\sin2\theta(t,E_j)\sin2\theta(t,E_m)}{2\eta'(t,E_m)(1+t)}dt.\nonumber
\end{align}
By the assumption \textbf{A1} one has for any $j\in I$ that $n_j>L(k_j)$. Then from \eqref{deLk}, \eqref{dedeltak},
\eqref{omegavarphiestimate} and \eqref{deAj}, one can obtain $C_j=400A_j< 600n_j\pi$. For any $\tilde{x}\geq T_j$ (assume $T_j\leq T_{m'}\leq \tilde{x}<T_{m'+1}$), since (from \eqref{detj}) $$T_{m'}>\sum_{i=1}^{m'}\left(\frac{2}{B_i\varepsilon_i}+\frac{2}{B_i{\delta}(k_i)}\right)\sum_{l=1}^{m'}C_l>\frac{2}{B_j{\delta}(k_j)}\sum_{l=1}^{m'}C_l,$$
we have
\begin{align}
	\abs{\frac{V(\tilde{x})}{\eta'(\tilde{x},E_j)}}\leq \frac{\sum_{l=1}^{m'}C_l}{B_j\tilde{x}}\leq\frac{\sum_{l=1}^{m'}C_l}{B_jT_{m'}}\leq\frac{{\delta}(k_j)}{2}.\nonumber
\end{align}
Similarly, for any $\tilde{x}\geq T_m$, one has
\begin{align}
	\abs{\frac{V(\tilde{x})}{\eta'(\tilde{x},E_m)}}\leq\frac{{\delta}(k_m)}{2}.\nonumber
\end{align}
Hence we obtain \eqref{lemmatildedetalki} with $l=m,j$ for any $\tilde{x}\geq \max\{T_m,T_j\}$ (${\delta}(k_m)={\delta}(k_j)$). Then by Lemma \ref{lemmafinal5.6}, one has
\begin{align}
\abs{\int_{T_m}^x\frac{C_m\cos4\theta(t,E_m)}{4\abs{\eta'(t,E_m)}(1+t)}dt}\leq \frac{1500\pi+1500\pi{\delta}(k_m)}{n_m}\ln x+\frac{O(n_m)}{T_m},\nonumber
\end{align}
and for $j<m,$
\begin{align}
	\abs{\int_{T_m}^x\frac{C_j\sin2\theta(t,E_j)\sin2\theta(t,E_m)}{2\eta'(t,E_m)(1+t)}dt}\leq \frac{(3000\pi+3000\pi{\delta}(k_m))n_j}{n_m\abs{n_m-n_j}}\ln x+\frac{O(n_j)}{T_m},\nonumber
\end{align}
for $j>m$, 
\begin{align}
	\abs{\int_{T_j}^x\frac{C_j\sin2\theta(t,E_j)\sin2\theta(t,E_m)}{2\eta'(t,E_m)(1+t)}dt}\leq \frac{(3000\pi+3000\pi{\delta}(k_m))n_j}{n_m\abs{n_m-n_j}}\ln x+\frac{O(n_j)}{T_j}.\nonumber
\end{align}
Therefore, one has
\begin{align}
	P_1+\sum_{\substack{j\neq m,j\leq M\\j\in I}}(-1)^{n_j}P_2(j)\leq& -100 \ln x +O\left(\sum_{j=1}^M\frac{n_j}{T_j}\right)+(3000\pi+3000\pi{\delta}(k_m))\nonumber\\
	&\times\left(\frac{1}{2n_m}+\sum_{\substack{j\neq m,j\leq M\\j\in I}}\frac{n_j}{n_m\abs{n_j-n_m}}\right)\ln x.\nonumber
\end{align}
From \eqref{detj} we have $\sum\frac{n_j}{T_j}< \infty$. Then by \eqref{conditionb30}, \eqref{lnrp1p2p3tmem} and \eqref{p3O1} one can obtain $R(\cdot,E_m)\in L^2(0,\infty)$.

Here we finish the proof of the Theorem \ref{theoremmain224}.
\end{proof}

\section{Appendix}

\begin{proof}[\textbf{Proof of Lemma \ref{L1}}]
	We only give the proof of \eqref{sinthetatovert}, the proof of \eqref{costhetatovert} is similar and omitted.
	
	Let $i_0$ be the largest integer such that $2\pi i_0<\theta(x_0)$. Then there exist $t_0<t_1<t_2<\cdots<t_N\leq x< t_{N+1}$, such that 
	\begin{align}
		\theta(t_i)=2\pi i_0+i\pi,\ i=0,1,\cdots,N+1.\nonumber
	\end{align}
	For any $i=0,1,\cdots,N$, one has that 
$
		\pi=\int_{t_i}^{t_{i+1}}\theta'(t)dt.$
	Since $\frac{C}{t_i^\beta}<\frac{1}{100}$, by \eqref{abstheta'xminu1lessthanc}
	one has
	\begin{align}
		\pi -\frac{2\pi C}{t_i^\beta}	\leq t_{i+1}-t_i\leq \pi +\frac{2\pi C}{t_i^\beta},\label{ti+1-ti}
	\end{align}
	and 
	\begin{align}
		t_i\geq t_0+i\frac{\pi}{2}.\label{tigeqt0plus}
	\end{align}
	For any $t\in [t_i,t_{i+1}]$, by \eqref{abstheta'xminu1lessthanc} and \eqref{ti+1-ti}, we can obtain
	\begin{align}
		2\pi i_0+i\pi+t-t_i -\frac{2\pi C}{t_i^\beta}\leq 	\theta(t)\leq 2\pi i_0+i\pi+t-t_i+\frac{2\pi C}{t_i^\beta}.
	\end{align}
	Therefore, by 
$
\int_{t_i}^{t_{i+1}}\abs{\sin\theta(t)}dt=\int_0^{t_{i+1}-t_i}\abs{\sin\theta(t+t_i)}dt,$
	we have
	\begin{align}
		\int_{t_i}^{t_{i+1}}\abs{\sin\theta(t)}dt\leq \int_0^\pi\abs{\sin\theta(t+t_i)}dt+\frac{2\pi C}{t_i^\beta}
		\leq2+\frac{10\pi C}{t_i^\beta},\label{twoplus}
	\end{align}
	and
	\begin{align}
		\int_{t_i}^{t_{i+1}}\abs{\sin\theta(t)}dt\geq 2-\frac{10\pi C}{t_i^\beta}.\label{twominus}
	\end{align}
	Since
	\begin{align}
		\int_{x_0}^x\frac{\sin\theta(t)}{t}dt=\int_{x_0}^{t_0}\frac{\sin\theta(t)}{t}dt+\sum_{i=0}^{N-1}\int_{t_i}^{t_{i+1}}\frac{\sin\theta(t)}{t}dt+\int_{t_N}^x\frac{\sin\theta(t)}{t}dt,\nonumber
	\end{align}
	and $\sin\theta(t)$ changes sign at each $t_i$, we know that the integral also has some cancellation between $(t_i,t_{i+1})$ and $(t_{i+1},t_{i+2})$. Then by \eqref{twoplus} and \eqref{twominus} we have
	\begin{align}
		\abs{	\int_{x_0}^x\frac{\sin\theta(t)}{t}dt}
		\leq& \frac{2\pi}{x_0}+\abs{\sum_{i=0}^{N-1}\int_{t_i}^{t_{i+1}}\frac{\sin\theta(t)}{t}dt}\nonumber\\
		\leq&\frac{2\pi}{x_0}+\abs{\sum_{i=0}^{N-1}  \int_{t_i}^{t_{i+1}}\frac{\sin\theta(t)}{t_i}dt}+\abs{\sum_{i=0}^{N-1} \int_{t_i}^{t_{i+1}}\left(\frac{\sin\theta(t)}{t}-\frac{\sin\theta(t)}{t_i}\right)dt}\nonumber\\
		\leq &\frac{2\pi}{x_0}+2\sum_{i=0}^{N-1} \frac{(-1)^i}{t_i}+\sum_{i=0}^{N-1}\frac{10\pi C}{t_i^{1+\beta}}+\sum_{i=0}^{N-1}\frac{4\pi}{t_i^2}.   \nonumber
	\end{align}
	Then by \eqref{tigeqt0plus} one can obtain
	\begin{align}
		\sum_{i=0}^{N-1}\frac{10\pi C}{t_i^{1+\beta}}+\sum_{i=0}^{N-1}\frac{4\pi}{t_i^2}\leq& \sum_{i=0}^{N-1}\left(\frac{10\pi C}{(t_0+i\frac{\pi}{2})^{1+\beta}}+\frac{4\pi}{(t_0+i\frac{\pi}{2})^2}\right)\nonumber\\
		\leq& \int_0^\infty\frac{10\pi C}{(t_0-\frac{\pi}{2}+\frac{\pi}{2}x)^{1+\beta}}+\frac{4\pi}{(t_0-\frac{\pi}{2}+\frac{\pi}{2}x)^2}dx\nonumber\\
		=&\frac{20C}{\beta}\frac{1}{(t_0-\frac{\pi}{2})^\beta}+\frac{8}{t_0-\frac{\pi}{2}}\nonumber\\
		\leq &\frac{\frac{30C}{\beta}+10}{x_0^\beta}.\nonumber
	\end{align}
	Thus, one has
	\begin{align}
		\abs{\int_{x_0}^x\frac{\sin\theta(t)}{t}dt}\leq \frac{\frac{30C}{\beta}+10\pi}{x_0^\beta}.\nonumber
	\end{align}
\end{proof}

	\noindent\textbf{Acknowledgments}\ The authors wish to thank Prof. Jifeng Chu of Hangzhou Normal University for useful comments on the earlier version of the manuscript.
	 The authors are supported by the National Natural
	Science Foundation of China (11871031).\
	
	\
	
	\noindent\textbf{Data Availability}\ Data sharing not applicable to this article as no datasets were generated or analysed during the current study.

	\
	
	
	\section*{Declarations}
	
	\noindent \textbf{Conflict of interest}
	On behalf of all authors, the corresponding author states that there is no conflict of interest.

\scriptsize	\bibliographystyle{abbrv} 
	\bibliography{reference}

\end{document}